\documentclass[journal,twoside,web]{ieeecolor}
\usepackage{generic}
\usepackage{cite}
\usepackage{amsmath,amssymb,amsfonts}
\usepackage{algorithmic}
\usepackage{graphicx}
\usepackage{algorithm,algorithmic}
\usepackage{hyperref}
\usepackage{comment}
\usepackage{textcomp}

\usepackage{amsthm}
\usepackage{arydshln}
\usepackage{empheq}

\newtheorem{theorem}{Theorem}
\newtheorem{assumption}{Assumption}

\newtheorem{lemma}[theorem]{Lemma}

\newtheorem{remark}{Remark}

\DeclareMathOperator{\rank}{rank}  
  
\DeclareMathOperator{\vecop}{vec}
\DeclareMathOperator{\col}{col}
\newcommand{\RR}{\mathbb{R}}
\newcommand{\CC}{\mathbb{C}}
\newcommand{\ZZ}{\mathbb{Z}}
\newcommand{\ZZn}{\mathbb{Z}_{\geq 0}}
\newcommand{\ZZp}{\mathbb{Z}_{> 0}}
\newcommand{\ie}{\textit{i.e.}}
\newcommand{\eg}{\textit{e.g.}}
\newcommand{\image}{\text{Im}}

\def\BibTeX{{\rm B\kern-.05em{\sc i\kern-.025em b}\kern-.08em
    T\kern-.1667em\lower.7ex\hbox{E}\kern-.125emX}}
\begin{document}
\title{One Equation to Rule Them All---Part I: Direct Data-Driven Cascade Stabilisation}
\author{Junyu Mao, \IEEEmembership{Student Member, IEEE}, Emyr Williams, \IEEEmembership{Student Member, IEEE}, 
Thulasi Mylvaganam, \IEEEmembership{Senior Member, IEEE}, and Giordano Scarciotti, \IEEEmembership{Senior Member, IEEE}
\thanks{J. Mao, E. Williams, and G. Scarciotti are with the Department of Electrical and Electronic Engineering, Imperial College London, SW7 2AZ London, U.K. T. Mylvaganam and E. Williams are with the Department of Aeronautics, Imperial College London, SW7 2AZ London, U.K.
(e-mail:  junyu.mao18@ic.ac.uk; emyr.williams18@ic.ac.uk; t.mylvaganam@ic.ac.uk; g.scarciotti@ic.ac.uk). }}

\maketitle

\begin{abstract}
In this article we present a framework for direct data-driven control for general problems involving interconnections of dynamical systems. We first develop a method to determine the solution of a Sylvester equation from data. Such solution is used to describe a subspace that plays a role in a large variety of problems. We then provide an error analysis of the impact that noise has on this solution. This is a crucial contribution because, thanks to the interconnection approach developed throughout the article, we are able to track how the noise propagates at each stage, and thereby provide bounds on the final designs. Among the many potential problems that can be solved with this framework, we focus on three representatives: cascade stabilisation, model order reduction, and output regulation. This manuscript studies the first problem, while the companion Part II addresses the other two. For each of these settings we show how the problems can be recast in our framework. In the context of cascade stabilisation, we consider the 2-cascade problem, the effect of noise through the cascade, as well as $\mathcal{N}$-cascade case, and we demonstrate that our proposed method is data efficient. The proposed designs are illustrated by means of a numerical example.
\end{abstract}

\begin{IEEEkeywords}
Control design, data-driven control, learning systems, linear matrix inequalities, robust control, output regulation, model order reduction
\end{IEEEkeywords}

\section{Introduction}
\label{sec:introduction}
\IEEEPARstart{T}{he} \emph{Sylvester equation}
\begin{equation}
\label{Eq:Sylvester}
\mathbf{A} \mathbf{X} - \mathbf{X} \mathbf{B} = \mathbf{CD},
\end{equation}
is a linear matrix equation in the unknown matrix $\mathbf{X}$, with 
$(\mathbf{A},\mathbf{B},\mathbf{C},\mathbf{D})$ given matrices of appropriate dimensions. This equation plays a fundamental role in numerous system analysis and control problems. An incomplete list of the system-theoretic applications of \eqref{Eq:Sylvester} include output regulation \cite{francis1976internal}, cascade stabilisation \cite{astolfi2022harmonic}, observer design \cite{luenberger1964observing}, eigenstructure assignment \cite{shafai1988algorithm}, model order reduction \cite{gallivan2004sylvester, astolfi2010model}, disturbance decoupling \cite{syrmos2002disturbance}, and hierarchical control \cite{girard2009hierarchical}. The reader is referred to \cite{astolfi2024role} for a unified overview of the applications of \eqref{Eq:Sylvester}, and to the more recent \cite{simpson-porco-arxiv} for a study of the properties of certain secondary operators associated to \eqref{Eq:Sylvester} upon which system analysis and design methods depend. 

A common feature across the aforementioned applications for which the Sylvester equation \eqref{Eq:Sylvester}  plays a central role, is that its
solution $\mathbf{X}$ often acts as an intermediate object, suggesting a certain change of state coordinates that enables a follow-up design. In the classic \textit{model-based control paradigm}, this solution $\mathbf{X}$ is computed based on the exact knowledge of $(\mathbf{A},\mathbf{B},\mathbf{C},\mathbf{D})$, which requires \textit{a priori} identification of the involved systems. This paradigm, however, is restrictive when only partial, uncertain or even no prior information about the involved systems is available, significantly obstructing the construction of practical solutions to the aforementioned problems. 

Inspired by Willems \textit{et al.}'s \textit{fundamental lemma} \cite{willems-lemma}, which shows that any possible trajectory of a linear system can be parameterised as a linear combination of collected trajectories under certain persistence of excitation conditions, a growing volume of research has been devoted to direct data-driven control design. In this setting, control laws are formulated directly with system data rather than by using transfer-function or state-space representations, thereby bypassing the need for an intermediate system identification step. A seminal contribution by De Persis and Tesi \cite{de2019formulas} provides formulas for data-driven control,  which are used for stabilisation and optimal control, exploiting data-dependent linear matrix inequalities (LMIs). This philosophy has been referred to as the \textit{direct data-driven control paradigm} and has inspired numerous model-free control solutions in various settings, \eg, data informativity \cite{informativity}, time-varying systems \cite{nortmann2023direct}, time-delay systems \cite{rueda2021data}, and nonlinear systems \cite{de2023learning}.

Following this direct data-driven paradigm, in this article we find the solution $\mathbf{X}$ to \eqref{Eq:Sylvester} by leveraging system data to replace the need of knowing the associated matrices. Building on this preliminary result, we introduce an interconnection-based framework for addressing classes of data-driven control problems. The central idea is to decompose complex tasks into interconnections of two (or multiple) dynamical systems, offering a unifying perspective that departs from existing literature. For instance, while specific problems such as data-driven output regulation or data-driven model order reduction have been studied in the literature, our approach is embedded in a broader theoretical context that enables a systematic treatment of multiple tasks using a shared methodology. Although we illustrate our framework by solving three representative problems, namely cascade stabilisation, model order reduction, and output regulation, its applicability extends beyond these cases. A list of potential additional applications of our framework can be inferred from \cite{astolfi2024role,simpson-porco-arxiv}. These have not been developed here due to space constraints.
A second key contribution lies in our comprehensive noise analysis, which is made possible precisely due to the interconnection-based formulation. Our framework enables the propagation and quantification of noise effects through multiple interconnected components, allowing us to derive explicit performance bounds on the final closed-loop systems directly from noise characteristics in the input data. This capability distinguishes our method from traditional two-step approaches, such as system identification followed by controller design, where such end-to-end noise guarantees are generally unavailable. The ability to track noise through the system interconnections, particularly in tasks like data-driven output regulation and model order reduction, represents a significant advantage and provides a further justification for the choice of a direct data-driven approach over two-step approaches.  

Below we provide an overview of the key related works in the first of the three selected application domains, \ie, cascade stabilisation. 

{\small \sf \color{nblue} \textit{Cascade Stabilisation:}} Stabilisation of subsystems connected in series is of significance in both theory and practice. In contrast to the design that considers the entire interconnected system as a unity, cascade stabilisation approaches follow a \textit{divide-and-conquer} philosophy, namely decomposing the overall high-dimensional stabilisation problem into low-dimensional, simpler sub-problems and stabilising subsystems recursively by exploiting their (series) interconnection structure. Successful attempts are given, \eg, in \cite{syrmos2002output, astolfi2022harmonic, chen2023data}. In particular, the approach proposed in \cite{astolfi2022harmonic} is known as the \textit{forwarding method}, where the solution of a Sylvester equation in the form of \eqref{Eq:Sylvester} is used to ensure that the latter subsystem strictly tracks a certain invariant subspace of the former one, guaranteeing the stability sequentially. However, these methods still rely on the models of each subsystem. Their data-driven counterparts, to the best of the authors' knowledge, are still missing. 

{\small \sf \color{nblue} \textbf{Contributions.}} The main contributions of this manuscript are summarised as follows.
\begin{itemize}
    \item[(I)] We derive a novel data-based reformulation of the Sylvester equation. This allows for the direct computation of its solution from data by solving a linear program (LP). For the practical case of noisy data, we establish a bound on the perturbation of the computed solution as a function of the noise level.
    
    \item[(II)] We leverage the framework (I) to develop a data-driven procedure that addresses the $2$-cascade stabilisation problem through solving data-dependent LPs and LMIs. 

    \item[(III)] For the setting in (II), we propose a noise-robust procedure that guarantees an input-to-state stability-like result in the presence of noise.

    \item[(IV)] We extend the proposed data-driven stabilisation approach to the general $\mathcal{N}$-cascade setting, and we show that this procedure requires a substantially smaller number of data samples than an alternative approach that treats the cascade as a monolithic system.
    
\end{itemize}

A preliminary small portion of this work has been accepted for presentation at the 2025 Conference on Decision and Control~\cite{Williams2025Direct}. That submission addresses solely the noise-free data-driven 2-cascade stabilisation problem from a domain-specific angle. In contrast, this manuscript provides significantly broader contributions, as outlined above. Furthermore, it forms Part I of a two‑part article. Part II \cite{mao2025partTwo} demonstrates the breadth and applicability of this framework by means of the model order reduction and output regulation problems.

{\small \sf \color{nblue} \textbf{Organisation.}} Section~\ref{sec:preliminaries} recalls some preliminaries needed for the subsequent sections. In Section~\ref{sec:dataSolutionSyl}, we introduce a framework for the computation of the solution to the Sylvester equation directly from 
data samples. This is then extensively exploited to produce direct data-driven approaches for the problem of cascade stabilisation in Section~\ref{sec:cascadeStabilisation}. Specifically, Section \ref{sec:cascade-full-info} proposes a data-driven procedure for the noise-free $2$-cascade setting, while Section \ref{sec:cascade-noise} investigates the influence of the noise. The results for the $2$-cascade stabilisation problem are then extended to a general setting in Section \ref{sec:N-cascade} where the cascade consists of $\mathcal{N}$ subsystems. The proposed design is demonstrated using an illustrative numerical example in Section \ref{sec:cascade-example}. Section~\ref{sec:conclusion} contains some concluding remarks.

{\small \sf \color{nblue} \textbf{Notation.}} We use standard notation. $\RR$, $\CC$ and $\ZZ$ denote the sets of real numbers, complex numbers and integer numbers. $\ZZn$ ($\ZZp$) denotes the set of nonnegative (positive) integers. The symbols $I$ and $\textbf{0}$ denote the identity matrix and the zero matrix, respectively, whose dimensions can be inferred from the context. $A^\top$ and  $\rank (A)$ indicate the transpose and the rank of any matrix $A$, respectively. The set of eigenvalues (singular values) of a matrix $A$ is denoted by $\lambda(A)$ ($\sigma(A)$). 
Given a matrix $A$ of full row rank, $A^\dagger$ represents the right inverse such that $AA^\dagger = I$. 
Given a symmetric matrix $A \in \mathbb{R}^{n \times n}$, $A \succ 0$ ($A \prec 0$) denotes that $A$ is positive- (negative-) definite. Given a matrix $A \in \mathbb{R}^{n \times m}$, the operator $\vecop(A)$ indicates the vectorization of $A$, which is the $nm \times 1$ vector obtained by stacking the columns of the matrix $A$ one on top of the other. $\|A\|_2$, $\|A\|_F$ and $\|A\|_\infty$ denote the spectral, Frobenius, and infinity norms of the matrix $A$, respectively. $\image(A)$ denotes the image of any matrix $A$. Given matrices  $X_1, \cdots, X_n$ (with the same number of columns), $\text{col}(X_1, \cdots, X_n)$ denotes their vertical concatenation. Given a vector $x$, the symbol $\|x\|_2$ ($\|x\|_F$) denotes its Euclidean (Frobenius) norm. 
The symbol $\otimes$ indicates the Kronecker product. The symbol $\iota$ denotes the imaginary unit. \\
We use capital versions of lower case letters to indicate the corresponding data matrices. For example, given a signal $x: \ZZn \to \RR^{n}$ and a positive integer $T \in \ZZp$, we define 
$$
\begin{array}{rl}
X_{-} &\!\!\!\!\!:=  \begin{bmatrix} 
    x(0) & x(1) & \cdots & x(T-1)
\end{bmatrix},\\[2mm]
X_{+} &\!\!\!\!\!:=  \begin{bmatrix} 
    x(1) & x(2) & \cdots & x(T)
\end{bmatrix}.
\end{array}
$$

\section{Preliminaries} \label{sec:preliminaries}
In this section, we recall two technical preliminary results, on which the main contributions of this article are based. First, in Section \ref{sec:preliminaries-sylvester}, we recall the role of the Sylvester equation in the context of the cascade interconnection of two discrete-time linear time-invariant (LTI) systems. Second, in Section~\ref{sec:preliminaries-data-driven}, we revisit the methodology outlined in \cite{de2019formulas} for the data-driven stabilisation of LTI systems.


\subsection{Sylvester Equation and Cascade Interconnection}
\label{sec:preliminaries-sylvester}
To avoid potential notational confusion\footnote{This is because the matrices $\mathbf{A}, \mathbf{B}, \mathbf{C}, \mathbf{D}$ and $\mathbf{X}$ possess some conventional interpretations in the state-space representation.} arising from \eqref{Eq:Sylvester}, throughout the remainder of this article, we consider the following Sylvester equation without loss of generality
\begin{equation} \label{eq:SylEqGeneric}
    \mathbf{A_2} \mathbf{\Theta} -  \mathbf{\Theta} \mathbf{A_1} = - \mathbf{B_2} \mathbf{C_1},
\end{equation}
in the unknown $\mathbf{\Theta} \in \mathbb{R}^{n_2\times n_1}$,
with $\mathbf{A_1} \in \RR^{n_1 \times n_1}$, $\mathbf{A_2} \in \RR^{n_2 \times n_2}$, $\mathbf{C_1} \in \RR^{p_1 \times n_1}$ and $\mathbf{B_2} \in \RR^{n_2 \times p_1}$. If $\mathbf{A}_1$ and $\mathbf{A}_2$ do not share eigenvalues (\ie{}, $\lambda(\mathbf{A_1}) \cap \lambda(\mathbf{A_2}) = \emptyset$), the Sylvester equation \eqref{eq:SylEqGeneric} has the unique solution $\mathbf{\Theta}$. 
Consider 
then two discrete-time LTI systems described by
\begin{equation} \label{Eq:SigmaSigma1}
\mathbf{\Sigma_1}:\,\,\begin{cases}\,\begin{aligned}
x_1(k+1) &= \mathbf{A_1} x_1(k) + \mathbf{B_1} u_1(k) \\
y_1(k) &= \mathbf{C_1} x_1(k) 
\end{aligned}\end{cases} 
\end{equation}
and
\begin{equation}
\label{Eq:SigmaSigma2}
\mathbf{\Sigma_2}:\,\,\begin{cases}\,\begin{aligned}
x_2(k+1) &= \mathbf{A_2} x_2(k) + \mathbf{B_2} u_2(k) \\
y_2(k) &= \mathbf{C_2} x_2(k) 
\end{aligned}\end{cases} 
\end{equation}
with $x_1(k) \in \RR^{n_1}$, $u_1(k) \in \RR^{m_1}$, $y_1(k) \in \RR^{p_1}$,
 $x_2(k) \in \RR^{n_2}$, $u_2(k) \in \RR^{m_2}$, and $y_2(k) \in \RR^{p_2}$. 
 
The Sylvester equation \eqref{eq:SylEqGeneric} is known to be associated with the cascade interconnection of $\mathbf{\Sigma}_1$ and $\mathbf{\Sigma}_2$, in which $\mathbf{\Sigma}_1$ drives $\mathbf{\Sigma}_2$ via $u_2(k) = y_1(k)$. In particular, when $u_1 \equiv 0$, the solution $\mathbf{\Theta}$ characterises an \textit{invariant subspace}
$$
\mathcal{M} =\{(x_1,x_2) \in \RR^{n_1 + n_2}: x_2 = \mathbf{\Theta} x_1\}, 
$$
of the cascade system. 
Thus, when $u_1 \not\equiv 0$, $\mathbf{\Theta}$ also suggests the construction of an error signal (or, equivalently, a change of coordinates). Namely consider 
$\zeta := x_2 - \mathbf{\Theta} x_1$, and note that 
\begin{equation*} 
\begin{aligned}
    \zeta(k+1) &\!=\!  x_2(k+1) - \mathbf{\Theta} x_1(k+1) \\
    &\!=\!   \mathbf{A_2} x_2(k) + \mathbf{B_2} u_2(k) 
     - \mathbf{\Theta} \left( \mathbf{A_1} x_1(k) + \mathbf{B_1} u_1(k)\right) \\
    &\!=\!    \mathbf{A_2} x_2(k) \!+\! \mathbf{B_2} \mathbf{C_1} x_1(k) 
     \!-\! \mathbf{\Theta} \!\left( \mathbf{A_1} x_1(k) \!+\! \mathbf{B_1} u_1(k)\!\right) \\
    &\!=\!  \mathbf{A_2} x_2(k) + \underbrace{(\mathbf{B_2} \mathbf{C_1}
     - \mathbf{\Theta} \mathbf{A_1})}_{=- \mathbf{A_2} \mathbf{\Theta} \text{ by \eqref{eq:SylEqGeneric}}} x_1(k) - \mathbf{\Theta} \mathbf{B_1} u_1(k) \\
    &\!=\!  \mathbf{A_2} (x_2(k) - \mathbf{\Theta} x_1(k)) - \mathbf{\Theta} \mathbf{B_1} u_1(k) \\
    &\!=\!  \mathbf{A_2} \zeta(k) - \mathbf{\Theta} \mathbf{B_1} u_1(k). 
\end{aligned}
\end{equation*}
At this point the literature studies two different cases, depending on whether $\mathbf{\Sigma_1}$ is considered the ``main system'' (\eg, a system to be controlled) while $\mathbf{\Sigma_2}$ is considered an ``auxiliary system'' (\eg, an internal model), or \textit{vice versa}. In the first case, which appears in problems such as cascade stabilisation, dynamic output regulation, and the so-called swapped moment matching method, $\mathbf{\Sigma_1}$ represents the (or a) ``main system'' with $u_1 \not \equiv 0$ \cite{simpson-porco-arxiv}. In the second case, which appears for instance in static output regulation and direct moment matching, $\mathbf{\Sigma_1}$ represents an ``auxiliary system'', typically an autonomous signal generator, and thus  $u_1 \equiv 0$ \cite{simpson-porco-arxiv}. In this Part I, we consider the first case, which will be revisited in Section~\ref{sec:cascade-problem-formulation}.

\subsection{Data-driven Stabilisation of Linear Systems}
\label{sec:preliminaries-data-driven}

Consider a discrete LTI system in the form \eqref{Eq:SigmaSigma1}. Suppose that the matrices $\mathbf{A_1}$ and $\mathbf{B_1}$ are unknown, but that data-collection experiments can be conducted to measure the values of $x_1(k)$ for given inputs $u_1(k)$. 
These input and state data collected from the experiment are used to form data matrices $U_{1,-}, X_{1,-}$, and $X_{1,+}$, each
of length $T$ (see the notation paragraph in Section \ref{sec:introduction} for the meaning of these symbols). Provided that the rank condition\footnote{The rank condition implies that $T \ge n_1 + m_1$, which is a lower bound on the required number of data points.}
\begin{equation} \label{eq:rank-condition}
    \text{rank}\left(\begin{bmatrix} X_{1,-} \\ U_{1,-}
    \end{bmatrix}\right) = n_1 + m_1,
\end{equation}
holds, a control law can either be designed after a system identification step, or directly from data \cite{willems-lemma}. In particular, if the above rank condition is satisfied, the \textit{open-loop} system can be represented in terms of the collected data as \cite{de2019formulas}
\begin{equation} \label{eq:data-openloop-representation}
    \begin{bmatrix}
        \mathbf{A_1} & \mathbf{B_1}
    \end{bmatrix} = X_{1,+}\begin{bmatrix}
        X_{1,-}\\U_{1,-}
    \end{bmatrix}^\dagger.
\end{equation}
Using \eqref{eq:data-openloop-representation},  the \textit{closed-loop} system $\mathbf{A_1}+\mathbf{B_1}K$ can be represented through data as
\begin{equation} \label{eq:linear-closed-loop-x1g}
\begin{split}    
    &\mathbf{A_1}+\mathbf{B_1}K 
    = X_{1,+}G_{K},
\end{split}
\end{equation}
where $G_{K} \in \RR^{T \times n_1}$ is any matrix such that
\begin{equation} \label{eq:K-to-g-mapping}
    \begin{bmatrix}
        I \\ K
    \end{bmatrix}=\begin{bmatrix}
        X_{1,-} \\ U_{1,-}
    \end{bmatrix}G_{K},
\end{equation}
as shown in \cite{de2019formulas}. Based on \eqref{eq:K-to-g-mapping}, in \cite{de2019formulas} it has also been shown 
that a stabilising feedback gain $K \in \mathbb{R}^{m_1\times n_1}$ can be designed via the solution of a convex optimisation problem in the form of a linear matrix inequality (LMI). Namely, $K := U_{1,-}Q_{x}(X_{1,-}Q_{x})^{-1}$, where $Q_{x} \in \mathbb{R}^{T\times n_1}$ is any matrix satisfying
\begin{equation} \label{eq:basic-sdp}
    \begin{bmatrix}
        X_{1,-}Q_{x} & X_{1,+}Q_{x} \\ Q_{x}^\top X_{1,+}^\top & X_{1,-}Q_{x}
    \end{bmatrix} \succ 0,
\end{equation}
is such that system \eqref{Eq:SigmaSigma1} in closed-loop with $u_1(k)= K x_1(k)$ is asymptotically stable, \ie, $\mathbf{A_1}+\mathbf{B_1}K$ is rendered Schur. In such case, 
\begin{equation} \label{eq:basic-sdp-GK}
G_{K} = Q_{x}(X_{1,-}Q_{x})^{-1}. 
\end{equation}
In practice, one does not have to guarantee the rank condition \eqref{eq:rank-condition} on the collected data. As a matter of fact, the stabilisation (by state feedback) can be achieved as long as the LMI \eqref{eq:basic-sdp} is feasible, which corresponds to the so-called \textit{data-informativity} condition. Note that this informativity condition is weaker than the rank condition \eqref{eq:rank-condition}, hence it is not sufficient for system identification, providing an advantage of direct data-driven stabilisation over two-step procedures involving system identification as the first step. The reader is directed to \cite{informativity} for further details.

Now suppose that the data samples are corrupted by the measurement noise, that is, only the noisy measurement $\bar{x}_1(k): = x_1(k) + \Delta x_1(k)$, with $\Delta x_1(k)$ the measurement noise, is available. Then the closed-loop matrix 
$\mathbf{A_1}+\mathbf{B_1}K$ can be represented through data as
\begin{equation} \label{eq:linear-closed-loop-x1g-noise}
\begin{split}    
    &\mathbf{A_1}+\mathbf{B_1}K  = (\bar{X}_{1,+} + R_{1, -}) G_{K},
\end{split}
\end{equation}
where $R_{1, -} := \mathbf{A_1}X_{1,-} - X_{1,+}$ and $G_{K} \in \RR^{T \times n_1}$ is any matrix such that
\begin{equation} \label{eq:K-to-g-mapping-noise}
    \begin{bmatrix}
        I \\ K
    \end{bmatrix}=\begin{bmatrix}
        \bar{X}_{1,-} \\ U_{1,-}
    \end{bmatrix}G_{K},
\end{equation}
provided that the rank condition 
\begin{equation} \label{eq:rank-condition-noise}
    \text{rank}\left(\begin{bmatrix} \bar{X}_{1,-} \\ U_{1,-}
    \end{bmatrix}\right) = n_1 + m_1,
\end{equation}
holds. The representation \eqref{eq:linear-closed-loop-x1g-noise} can be exploited to form a noise-robust result, where a more conservative LMI is formulated to ensure 
$\mathbf{A_1}+\mathbf{B_1}K$ is Schur. Namely, provided that the following signal-to-noise ratio (SNR) condition
\begin{equation} \label{eq:snr-condition-normal}
    R_{1,-} R_{1,-}^\top \preceq \frac{\alpha^2}{2(2+\alpha)} \bar{X}_{1,+} \bar{X}_{1,+}^\top,
\end{equation}
holds, $K := U_{1,-}Q_{x}(\bar{X}_{1,-}Q_{x})^{-1}$, where $Q_{x} \in \mathbb{R}^{T\times n_1}$ is any matrix satisfying
\begin{subequations}
    \label{eq:sdp-for-noisy-system}
    \begin{align}
    &\begin{bmatrix} 
        \bar{X}_{1,-}Q_x - \alpha \bar{X}_{1,+} \bar{X}_{1,+}^\top& \bar{X}_{1,+}Q_x \\ Q_x^\top \bar{X}_{1,+}^\top & \bar{X}_{1,-}Q_x
    \end{bmatrix} \succ 0, \\
    &\begin{bmatrix} 
        I & Q_x \\ Q_x^\top & \bar{X}_{1,-} Q_x
    \end{bmatrix} \succ 0,
     \end{align}
\end{subequations}
is such that $\mathbf{A_1}+\mathbf{B_1}K$ is still rendered Schur.

\section{Direct Data-driven Solutions of the Sylvester Equation} \label{sec:dataSolutionSyl}
This section addresses the problem of directly computing the solution to Sylvester equations using data samples collected from the systems of interest, without resorting to intermediate system identification procedures. We investigate an instrumental (experimental) configuration, in which the dynamics of the first subsystem are known, whereas those of the second subsystem remain unknown. We will see across Part I and Part II \cite{mao2025partTwo} how such a configuration appears extensively in problems of practical interest. 
Furthermore, we analyse how the presence of measurement and process noise in these subsystems can impact the computed solution. Our noise analysis is constructed in such a way as to carry over to the following applications, offering insights into the robustness of the proposed methodologies under empirical conditions.

\subsection{Solving the Sylvester Equation from Noise-free Data} \label{sec:solutionPrimalSyl}
Consider that one knows the system matrices of $\mathbf{\Sigma_1}$ (\ie, $\mathbf{A_1}$, $\mathbf{B_1}$ and $\mathbf{C_1}$ are given), whereas the
system matrices of $\mathbf{\Sigma}_2$ 
are entirely unknown. 
\begin{remark} \label{remark:Sigma1known}
We stress from the outset that assuming that $\mathbf{\Sigma}_1$ is known is largely without loss of generality because, as it will become clear later in this article and the companion Part II \cite{mao2025partTwo}, this knowledge arises from either (i) the sequential nature of the proposed data-driven procedure, or from the fact that (ii) the first subsystem is designed by the user and does not have a physical embodiment. Nevertheless, for the sake of completeness, the general case where both subsystems are unknown is discussed in Part II.
\end{remark}
Suppose that we have access to exact measurements of $x_2(k)$ and $u_2(k)$.
For the time being, we introduce the following assumption on the data. 
\begin{assumption} \label{ass:rankXU}
The available data matrices are such that the following rank condition is satisfied
\begin{equation} \label{eq:rankXU}
    \rank \left(\begin{bmatrix}
        X_{2, -} \\ U_{2, -}
    \end{bmatrix}\right)
     = n_2 + m_2.
\end{equation} 
\end{assumption}
A sufficient condition for Assumption~\ref{ass:rankXU} to hold is that the input sequence is persistently exciting of order $n_2+1$, and that the pair $(\mathbf{A_2}, \mathbf{B_2})$ is controllable, see \cite[Corollary~2]{willems-lemma}. Under Assumption \ref{ass:rankXU} on the data matrices, the following result establishes a data-based reformulation of the Sylvester equation such that \eqref{eq:SylEqGeneric} can be solved directly from the collected data.

\begin{lemma} \label{lemma:computePI}
Consider systems $\mathbf{\Sigma_1}$ and $\mathbf{\Sigma_2}$, and the Sylvester equation \eqref{eq:SylEqGeneric}. 
Suppose that Assumption \ref{ass:rankXU} holds and that $\lambda(\mathbf{A_1}) \cap \lambda(\mathbf{A_2}) = \emptyset$. Then any matrix $G_\Theta \in \RR^{T \times n_1}$ that satisfies
\begin{subequations} \label{eq:dataRepPi} 
  \begin{empheq}[left=\empheqlbrace]{align}
    X_{2, +}  G_\Theta &= X_{2, -} G_\Theta \mathbf{A_1}   \label{eq:dataRepPi1} \\
    U_{2, -}  G_\Theta &= \mathbf{C_1}    \label{eq:dataRepPi2}
  \end{empheq}
\end{subequations}
is such that
\begin{equation} \label{eq:parameterizedPi}
    \mathbf{\Theta} := X_{2, -} G_\Theta
\end{equation}
is the solution of \eqref{eq:SylEqGeneric}.
Conversely, the solution of \eqref{eq:SylEqGeneric} can be written as in \eqref{eq:parameterizedPi}, with $G_\Theta$ solution of \eqref{eq:dataRepPi}.  
\end{lemma}

\begin{proof}
\textit{(Necessity)} Suppose that $\mathbf{\Theta}$ is a solution of the Sylvester equation \eqref{eq:SylEqGeneric}. Rewriting the equation in matrix form, we obtain
\begin{equation} \label{eq:SylPrimalMat}
    \begin{bmatrix}
        \mathbf{A}_2 & \mathbf{B}_2
    \end{bmatrix}
    \begin{bmatrix}
        \mathbf{\Theta} \\  \mathbf{C}_1 
    \end{bmatrix}
    = 
    \mathbf{\Theta} \mathbf{A}_1.
\end{equation}
By the rank condition \eqref{eq:rankXU}, and applying the Rouché–Capelli theorem, there exists a matrix \( G_\Theta \) such that
\begin{equation} \label{eq:parametrizationPi}
    \begin{bmatrix}
        \mathbf{\Theta} \\  \mathbf{C}_1 
    \end{bmatrix}
    = 
    \begin{bmatrix}
        X_{2,-} \\ U_{2,-}
    \end{bmatrix}
    G_\Theta.
\end{equation}
Substituting \eqref{eq:parametrizationPi} into \eqref{eq:SylPrimalMat} yields
\begin{align*}
    \begin{bmatrix}
        \mathbf{A}_2 & \mathbf{B}_2
    \end{bmatrix}
    \begin{bmatrix}
        X_{2,-} \\ U_{2,-}
    \end{bmatrix}
    G_\Theta
    = X_{2,+} G_\Theta
    = \mathbf{\Theta} \mathbf{A}_1
    = X_{2,-} G_\Theta \mathbf{A}_1,
\end{align*}
where we used the relation $X_{2,+} = \mathbf{A}_2 X_{2,-} + \mathbf{B}_2 U_{2,-}$ for the first equality, \eqref{eq:SylPrimalMat} for the second equality, and the first block row of \eqref{eq:parametrizationPi} for the third equality. This chain of equalities yields \eqref{eq:dataRepPi1}. In addition, the last block row of \eqref{eq:parametrizationPi} yields \eqref{eq:dataRepPi2} directly. 

\textit{(Sufficiency)} Suppose that $G_{\Theta}$ is a solution of \eqref{eq:dataRepPi}. It follows from \eqref{eq:dataRepPi1} and \eqref{eq:dataRepPi2} that
\begin{align*}
    \mathbf{A_2} (X_{2, -} G_\Theta)  - (X_{2, -} G_\Theta) & \mathbf{A_1} \overset{\eqref{eq:dataRepPi1}}{=} \mathbf{A_2} X_{2, -} G_\Theta -  X_{2, +}  G_\Theta \\
    &= (\mathbf{A_2} X_{2, -} - X_{2, +}) G_\Theta \\
    &= - \mathbf{B_2} U_{2, -}  G_\Theta 
    \overset{\eqref{eq:dataRepPi2}}{=} - \mathbf{B_2} \mathbf{C_1},
\end{align*}
where in the third equality we exploited that $X_{2,+} = \mathbf{A}_2 X_{2,-} + \mathbf{B}_2 U_{2,-}$, implying that $X_{2, -} G_\Theta$ is a solution of \eqref{eq:SylEqGeneric}. Furthermore, the condition $\lambda(\mathbf{A_1}) \cap \lambda(\mathbf{A_2}) = \emptyset$ implies that this solution (see \cite[Proposition 6.2]{antoulas2005approximation}) is unique.  
\end{proof}

\begin{remark}
The equations \eqref{eq:dataRepPi} can be used as constraints in a feasibility problem to determine a matrix \( G_\Theta \), and thereby construct the solution \(\mathbf{\Theta} = X_{2,-} G_\Theta\) of \eqref{eq:SylEqGeneric} directly from the available data. Since the constraints are linear in the decision variable \(G_\Theta\), this feasibility problem is a linear program (LP) and can be solved efficiently by modern solvers (\eg, MOSEK \cite{aps2019mosek}) in polynomial time \cite{potra2000interior}, scaling well to cases where $n_1, n_2$, and $T$ are on the order up to thousands.
\end{remark}

\begin{remark}[Data Informativity]
To guarantee the existence of $G_\Theta$ in \eqref{eq:parametrizationPi}, the sufficient condition \eqref{eq:rankXU} can be replaced by the weaker condition
\begin{equation} \label{eq:dataInformativityCondition}
        \image\left(\begin{bmatrix}
        \mathbf{\Theta} \\  \mathbf{C}_1 
    \end{bmatrix}\right) \subset \image\left(\begin{bmatrix}
        X_{2,-} \\ U_{2,-}
    \end{bmatrix}\right).
\end{equation}
This condition, which is equivalent to the feasibility of \eqref{eq:dataRepPi}, is not sufficient for the identification of $\mathbf{\Sigma_2}$. However, it is necessary and sufficient for determining the unique solution $\mathbf{\Theta}$ 
of \eqref{eq:SylEqGeneric} by solving \eqref{eq:dataRepPi}, \ie, \eqref{eq:dataInformativityCondition} or the feasibility of \eqref{eq:dataRepPi} serves as the data-informativity condition.
Analogously, all rank conditions that we provide in this article 
could be turned, mutatis mutandis, into weaker informativity conditions in terms of the feasibility of relevant matrix equalities or LMIs. However, for simplicity of exposition, we formulate only the rank conditions in the rest of the article. 
\end{remark}

\subsection{Error Analysis under Noise-corrupted Data} \label{sec:noiseErrorSyl}
In Sections~\ref{sec:solutionPrimalSyl}, we primarily focused on the idealised, noise-free scenario, where exact measurements (of the input and state), were available and the system dynamics were assumed to be unaffected by noise. However, in practical settings, noise may affect the system in several ways: (i) the measured state and input trajectories may be corrupted by sensor noise; (ii) the system may be influenced by both exogenous disturbances and internal process noise; (iii) the prior knowledge of $\mathbf{\Sigma}_1$ may not be exact. Motivated by this, in this section we provide an error analysis 
accounting for all these potential uncertainties, which will be instrumental to then carry out an error analysis in all subsequent applications relying on solutions of Sylvester equations.

Revisiting the configuration in Section~\ref{sec:solutionPrimalSyl}, consider a perturbed version 
of (the unknown) $\mathbf{\Sigma_2}$ given by
\begin{equation}  \label{eq:sysNoises}
     \quad x_2(k+1) = \mathbf{A_2} x_2(k) + \mathbf{B_2} u_2(k) + d_2(k),
\end{equation}
where $d_2(k) \in \mathbb{R}^{n_2}$ denotes unknown additive system noise, accounting for unmodeled dynamics and/or exogenous disturbances. Now suppose that only the following corrupted quantities are available:
\begin{subequations} \label{eq:uncertainty}
\begin{align}
    \bar{x}_2(k) &= x_2(k) + \Delta x_2(k), \\
    \bar{u}_2(k) &= u_2(k) + \Delta u_2(k), \\
    \mathbf{\bar{A}_1} &= \mathbf{A_1} + \Delta \mathbf{A_1},
\end{align}
\end{subequations}
with $\Delta x_2(k)$ and $\Delta u_2(k)$ the measurement noise on the state and input, and $\Delta \mathbf{A_1}$ the knowledge mismatch on $\mathbf{A_1}$.
To streamline the presentation, we begin by defining a compact or ``encapsulated'' noise signal as
\begin{equation*}
    r_2(k) := \mathbf{A_2} \Delta x_2(k) - \Delta x_2(k+1) + \mathbf{B_2} \Delta u_2(k) - d_2(k),
\end{equation*}
and, correspondingly, we define
\begin{equation} \label{eq:Rdef}
    R_{2,-} := \mathbf{A_2} \Delta X_{2,-} - \Delta X_{2,+} +\mathbf{B_2} \Delta U_{2,-} - D_{2,-}\,.
\end{equation}
which will play a central role in the analysis presented in this subsection. Exploiting this notation, the dynamics of the corrupted state can be compactly expressed as 
\begin{equation} \label{eq:xbarDynamics}
\begin{aligned} 
    \mathbf{\bar{\Sigma}_2}: \bar{x}_2(k+1) = \mathbf{A_2}  \bar{x}_2(k) + \mathbf{B_2} \bar{u}_2(k) - r_2(k).
\end{aligned}
\end{equation}
As before, we impose a rank condition on the data matrices, analogous to Assumption~\ref{ass:rankXU}.

\begin{assumption} \label{ass:rankX-U}
The available data matrices are such that the following rank condition is satisfied
\begin{equation} \label{eq:rankX-U}
    \rank \left(\begin{bmatrix}
        \bar{X}_{2, -} \\ \bar{U}_{2, -} 
    \end{bmatrix}\right)
     = n_2 + m_2.
\end{equation}
\end{assumption}

The following preliminary result presents a ``noisy counterpart'' of Lemma~\ref{lemma:computePI}, under the unrealistic assumption (that we will late remove) that the signal $r_2$ is available for the time being. 

\begin{lemma} \label{lemma:dataRepPibar}
Consider system $\mathbf{\bar{\Sigma}_2}$ defined in \eqref{eq:xbarDynamics} and the Sylvester equation \eqref{eq:SylEqGeneric}. Suppose that Assumption \ref{ass:rankX-U} holds and that $\lambda(\mathbf{A_1}) \cap \lambda(\mathbf{A_2}) = \emptyset$. Then any matrix $\bar{G}_\Theta \in \RR^{T \times n_1}$ that satisfies
\begin{subequations} \label{eq:dataRepPibar} 
  \begin{empheq}[left=\empheqlbrace]{align}
    (\bar{X}_{2, +} + R_{2, -})  \bar{G}_\Theta &= \bar{X}_{2, -} \bar{G}_\Theta \mathbf{A_1}    \label{eq:dataRepPibar1} \\
    \bar{U}_{2, -}  \bar{G}_\Theta &= \mathbf{C_1}    \label{eq:dataRepPibar2}
  \end{empheq}
\end{subequations}
is such that
\begin{equation} \label{eq:parameterizedPibar}
    \mathbf{\Theta} := \bar{X}_{2, -} \bar{G}_\Theta\,,
\end{equation}
is the solution of \eqref{eq:SylEqGeneric}.
Conversely, the solution of \eqref{eq:SylEqGeneric} can be written as in \eqref{eq:parameterizedPibar}, with $\bar{G}_\Theta$ solution of \eqref{eq:dataRepPibar}.  
\end{lemma}

\begin{proof}
Note that under the rank condition \eqref{eq:rankX-U}, now there exists a matrix $\bar{G}_\Theta$ such that
\begin{equation} \label{eq:parametrizationPibar}
    \begin{bmatrix}
        \mathbf{\Theta} \\ \mathbf{C_1}
    \end{bmatrix}
    = 
    \begin{bmatrix}
        \bar{X}_{2, -} \\ \bar{U}_{2, -}
    \end{bmatrix}
    \bar{G}_\Theta. 
\end{equation}
Exploiting the above parametrization,  \eqref{eq:SylPrimalMat} in the proof of Lemma \ref{lemma:computePI} becomes
\begin{align*} 
    \begin{bmatrix}
         \mathbf{A_2} & \mathbf{B_2}
    \end{bmatrix}
    \begin{bmatrix}
       \bar{X}_{2, -} \\ \bar{U}_{2, -}
    \end{bmatrix}
    \bar{G}_\Theta
    =& \;
    (\mathbf{A_2} \bar{X}_{2, -} + \mathbf{B_2} \bar{U}_{2, -}) \bar{G}_{\Theta} \\
    \overset{\eqref{eq:xbarDynamics}}{=} &\, (\bar{X}_{2, +} + R_{2, -}) \bar{G}_\Theta \\
    \overset{\eqref{eq:SylPrimalMat}}{=} &\,  
    \mathbf{\Theta}  \mathbf{A_1}  
    \overset{\eqref{eq:parametrizationPibar}}{=}
    \bar{X}_{2, -} \bar{G}_{\Theta}  \mathbf{A_1}. 
\end{align*}
By exploiting the above equation, the proofs for necessity and sufficiency can be constructed analogously to those in Lemma \ref{lemma:computePI}, and hence are omitted for brevity.
\end{proof}


In practical scenarios, the noise embedded in $R_{2,-}$ is typically unknown and not directly measurable, and the system $\mathbf{\Sigma_1}$ may not be known exactly. As a result, the \textit{nominal} formulation with constraints as in \eqref{eq:dataRepPibar} is unavailable. Instead, one is naturally led to consider an \textit{empirical} version of the problem: namely, to determine a matrix $\hat{G}_{\Theta} \in \RR^{T \times n_1}$ such that
\begin{subequations} \label{eq:dataRepPiPerturbed} 
  \begin{empheq}[left=\empheqlbrace]{align}
    \bar{X}_{2, +} \hat{G}_{\Theta} &= \bar{X}_{2, -} \hat{G}_{\Theta} \mathbf{\bar{A}_1} \label{eq:dataRepPiPerturbed1} \\
    \bar{U}_{2, -}  \hat{G}_{\Theta} &=\mathbf{C_1}\label{eq:dataRepPiPerturbed2}
  \end{empheq}
\end{subequations}
yielding the empirical solution 
\begin{equation} \label{eq:PiPerturbed}
    \mathbf{\hat{\Theta}} := \bar{X}_{2, -} \hat{G}_{\Theta}
\end{equation}
to \eqref{eq:SylEqGeneric}. Note that \eqref{eq:dataRepPiPerturbed}, in addition to omitting $R_{2, -}$, it also replaces $\mathbf{A}_1$ with $\mathbf{\bar{A}_1}$, implying that in practice we do not have perfect knowledge of $\mathbf{A}_1$. Obviously, if $R_{2,-} = \mathbf{0}$ and $\Delta \mathbf{A_1} = \mathbf{0}$, the empirical problem coincides with the nominal one, and thus $\mathbf{\hat{\Theta}} = \mathbf{\Theta}$. 
When $R_{2,-} \neq \mathbf{0}$ or $\Delta \mathbf{A_1} \neq \mathbf{0}$, it becomes essential to understand how these affect the accuracy of the empirical solution $\hat{\mathbf{\Theta}}$. The aim of such analysis is to provide an upper bound on the deviation between the empirical $\hat{\mathbf{\Theta}}$ and the exact solution $\mathbf{\Theta}$, expressed in terms of the norms of $R_{2,-}$ and $\Delta \mathbf{A_1}$. 

We highlight that analysing the deviation between $\bar{G}_{\Theta}$ and $\hat{G}_{\Theta}$ is inherently nontrivial, as neither $\bar{G}_{\Theta}$ nor $\hat{G}_{\Theta}$ are unique solutions to equations \eqref{eq:dataRepPibar} and \eqref{eq:dataRepPiPerturbed}, respectively. Consequently, the difference $\bar{G}_{\Theta} - \hat{G}_{\Theta}$ is not uniquely defined. In fact, its norm may be unbounded due to the possibility of arbitrary elements lying in the kernel space of the associated linear mappings being added to either solution. Despite this ambiguity in $\bar{G}_{\Theta}$ and $\hat{G}_{\Theta}$, we show that the empirical error of the solution, \ie, $\Delta \mathbf{\Theta} := \mathbf{\Theta} - \mathbf{\hat{\Theta}}$, is uniquely defined and bounded for any empirically evaluated $\hat{G}_{\Theta}$. To establish this, the following result demonstrates that the error $\Delta \mathbf{\Theta}$ can be characterised as the solution to a new Sylvester equation.

\begin{lemma} \label{lemma:errorSyl}
Suppose that the assumptions in Lemma \ref{lemma:dataRepPibar} hold. Let $\hat{G}_{\Theta}$ be a solution of \eqref{eq:dataRepPiPerturbed} and $\hat{\mathbf{\Theta}}$ be defined as in \eqref{eq:PiPerturbed}. 
Then, $\Delta \mathbf{\Theta}$ is the unique solution of the Sylvester equation
\begin{equation} \label{eq:errorSyl}
    \mathbf{A_2} \Delta \mathbf{\Theta} -  \Delta \mathbf{\Theta} \mathbf{A_1} = - R_{2, -} \hat{G}_{\Theta} - \bar{X}_{2, -} \hat{G}_{\Theta} \Delta \mathbf{A_1}. 
\end{equation}
\end{lemma}
\begin{proof}
We begin by analysing the perturbation between the nominal equations \eqref{eq:dataRepPibar} and the perturbed equations \eqref{eq:dataRepPiPerturbed}. Define the perturbation matrix \(\Delta G_\Theta := \bar{G}_\Theta - \hat{G}_\Theta\), and rewrite \eqref{eq:dataRepPiPerturbed} in terms of \(\Delta G_\Theta\) as
\begin{align*}
    \bar{X}_{2,+} (\bar{G}_\Theta - \Delta G_\Theta) &= \bar{X}_{2,-} (\bar{G}_\Theta - \Delta G_\Theta)(\mathbf{A_1} + \Delta \mathbf{A_1}), \\ 
    \bar{U}_{2,-} (\bar{G}_\Theta - \Delta G_\Theta) &= \mathbf{C_1}. 
\end{align*}
Subtracting the above equations from the nominal equations \eqref{eq:dataRepPibar} yields
\begin{subequations}\label{eq:dataRepDeltaPiAll} 
\begin{align}
    \bar{X}_{2,+} \Delta G_\Theta  &= -R_{2,-} \bar{G}_\Theta + \bar{X}_{2,-} \Delta G_\Theta \mathbf{A}_1 -\bar{X}_{2,-} \hat{G}_\Theta \Delta \mathbf{A_1}, \label{eq:dataRepDeltaPi1} \\
    \bar{U}_{2,-} \Delta G_\Theta &= \mathbf{0}. \label{eq:dataRepDeltaPi2}
\end{align}
\end{subequations}
Note that \eqref{eq:dataRepDeltaPiAll} and \eqref{eq:dataRepPiPerturbed} are linear equations with exactly the same coefficient matrix (associated with the respective unknowns). Consequently, if $\hat{G}_{\Theta}$ is not unique, so is $\Delta G_\Theta$.
However, our primary interest lies in the error $\Delta \mathbf{\Theta}$. In what follows we show that this, instead, is uniquely determined. 
Note that $\Delta \mathbf{\Theta} = \bar{X}_{2,-} \Delta G_\Theta$ and we obtain
\begin{align*}
    \mathbf{A}_2 \Delta & \mathbf{\Theta} - \Delta \mathbf{\Theta} \mathbf{A}_1 
    = \mathbf{A}_2 \bar{X}_{2,-} \Delta G_\Theta - \bar{X}_{2,-} \Delta G_\Theta \mathbf{A}_1 \\
    &= \left( \mathbf{A}_2 \bar{X}_{2,-} - \bar{X}_{2,+} \right) \Delta G_\Theta - R_{2,-} \bar{G}_\Theta -\bar{X}_{2,-} \hat{G}_\Theta \Delta \mathbf{A_1} \\
    &= \left( R_{2,-} - \mathbf{B}_2 \bar{U}_{2,-} \right) \Delta G_\Theta - R_{2,-} \bar{G}_\Theta -\bar{X}_{2,-} \hat{G}_\Theta \Delta \mathbf{A_1} \\
    &= R_{2,-} (\Delta G_\Theta - \bar{G}_\Theta) - \mathbf{B}_2 \bar{U}_{2,-} \Delta G_\Theta - \bar{X}_{2,-} \hat{G}_\Theta \Delta \mathbf{A_1} \\
    &= - R_{2,-} \hat{G}_\Theta - \bar{X}_{2,-} \hat{G}_\Theta \Delta \mathbf{A_1},
\end{align*}
where 
in the  second equality we used \eqref{eq:dataRepDeltaPi1}, in the third equality we used the identity \(\bar{X}_{2,+} = \mathbf{A}_2 \bar{X}_{2,-}  + \mathbf{B}_2 U_{2,-} - R_{2,-}\), implied by \eqref{eq:xbarDynamics}, 
and in the fifth equality we exploited \eqref{eq:dataRepDeltaPi2}. 
Hence, \(\Delta \mathbf{\Theta}\) satisfies the Sylvester equation \eqref{eq:errorSyl}.
Finally, since \(\lambda(\mathbf{A}_1) \cap \lambda(\mathbf{A}_2) = \emptyset\) by assumption, this Sylvester equation admits a unique solution. Therefore, although \(\Delta G_\Theta\) itself is not unique, the matrix \(\Delta \mathbf{\Theta}\) is uniquely determined.
\end{proof}


Exploiting the characterisation of the empirical error $\Delta \mathbf{\Theta}$ given in Lemma \ref{lemma:errorSyl}, we are ready to provide an upper bound on $\Delta \mathbf{\Theta}$ in terms of the noise terms $R_{2, -}$ and $\Delta \mathbf{A_1}$. 

\begin{theorem} \label{lemma:errorBound}
Suppose that the assumptions in Lemma \ref{lemma:dataRepPibar} hold. Let $\hat{G}_{\Theta}$ be a solution of \eqref{eq:dataRepPiPerturbed} and $\hat{\mathbf{\Theta}}$ be defined as in \eqref{eq:PiPerturbed}.  Then, the bound
\begin{equation} \label{eq:boundPi}
    \|\Delta \mathbf{\Theta}\|_2 \leq  \|\Delta \mathbf{\Theta}\|_F
    \leq \frac{\|\hat{G}_{\Theta}\|_F \left( \| R_{2, -} \|_2 + \|\bar{X}_{2,-}\|_2 \|\Delta \mathbf{A_1}\|_2\right)}{\sigma_{min}(I \otimes \mathbf{A_2} - \mathbf{A_1}^\top \otimes I)} 
\end{equation}
holds. 
\end{theorem}

\begin{proof}
For notational convenience, define \(\mathcal{A} := I \otimes \mathbf{A}_2 - \mathbf{A}_1^\top \otimes I\). By \cite[Proposition 6.2]{antoulas2005approximation} and Lemma~\ref{lemma:errorSyl}, the error \(\Delta \mathbf{\Theta}\) can be expressed as
\begin{equation} \label{eq:errorBoundSolution}
    \vecop(\Delta \mathbf{\Theta}) = 
    \mathcal{A}^{-1} \vecop(- R_{2,-} \hat{G}_\Theta - \bar{X}_{2,-} \hat{G}_\Theta \Delta \mathbf{A_1}).
\end{equation}
Applying the mixed-norm inequality\footnote{Given matrices \(\mathcal{A}\) and \(\mathcal{B}\) of compatible dimensions, the inequalities \(\|\mathcal{A} \mathcal{B}\|_F \leq \|\mathcal{A}\|_2 \|\mathcal{B}\|_F\) and \(\|\mathcal{A} \mathcal{B}\|_F \leq \|\mathcal{A}\|_F \|\mathcal{B}\|_2\) hold.} to \eqref{eq:errorBoundSolution} yields
$$
    \|\vecop(\Delta \mathbf{\Theta})\|_F 
    \leq \|\mathcal{A}^{-1}\|_2 \|\vecop(R_{2,-} \hat{G}_\Theta  + \bar{X}_{2,-} \hat{G}_\Theta \Delta \mathbf{A_1})\|_F.
$$
Recognising that \(\|\vecop( \cdot )\|_F = \| \cdot \|_F\) 
along with the fact that \(\|\mathcal{A}^{-1}\|_2 = 1/\sigma_{\min}(\mathcal{A})\), we obtain
$$
    \|\Delta \mathbf{\Theta}\|_F 
    \leq \frac{1}{\sigma_{\min}(\mathcal{A})} \|R_{2,-} \hat{G}_\Theta + \bar{X}_{2,-} \hat{G}_\Theta \Delta \mathbf{A_1}\|_F.
$$
Applying the triangle inequality of the Frobenius norm and again the mixed-norm inequality gives
\begin{align*}
    \|\Delta \mathbf{\Theta}\|_F &\leq \frac{1}{\sigma_{\min}(\mathcal{A})} \left( \|R_{2,-} \hat{G}_\Theta\|_F + \| \bar{X}_{2,-} \hat{G}_\Theta \Delta \mathbf{A_1}\|_F \right) \\
    &\leq \frac{1}{\sigma_{\min}(\mathcal{A})} \|\hat{G}_\Theta\|_F \left( \|R_{2,-}\|_2  + \| \bar{X}_{2,-} \|_2 \|\Delta \mathbf{A_1}\|_2\right),
\end{align*}
Last, recognising that $\|\Delta \mathbf{\Theta}\|_2 \leq \|\Delta \mathbf{\Theta}\|_F$ for any $\Delta \mathbf{\Theta}$ completes the proof.
\end{proof}

\begin{remark} \label{remark:minmize-norm-G}
In \eqref{eq:boundPi}, it can be seen that the shared bound constant across $R_{2,-}$ and $\Delta \mathbf{A}_1$ depends on two key quantities. The first is the minimal singular value of the matrix $ (I \otimes \mathbf{A}_2 - \mathbf{A}_1^\top \otimes I) $. While this eigenvalue is unknown, it remains constant for fixed systems. The second is the Frobenius norm of the matrix $ \hat{G}_\Theta $, which can be evaluated a posteriori upon solving the feasibility problem \eqref{eq:dataRepPiPerturbed}. Motivated by this relationship, and with the aim of minimising the error $\Delta \boldsymbol{\Theta}$, the bound suggests replacing the feasibility problem \eqref{eq:dataRepPiPerturbed} with the least-norm optimisation problem
\begin{equation} \label{eq:minNormOptimization}
    \begin{aligned}  
    \text{minimize} \quad &  \|\hat{G}_\Theta\|_{F} \\
    \text{subject to} \quad & 
    \eqref{eq:dataRepPiPerturbed1}, \eqref{eq:dataRepPiPerturbed2},
    \end{aligned} 
\end{equation}
which remains convex in the decision variable $\hat{G}_\Theta $. Furthermore, since the matrix $(I \otimes \mathbf{A}_2 - \mathbf{A}_1^\top \otimes I) $ is invariant with respect to the data, the norm $\|\hat{G}_\Theta\|_{F}$ serves as a proxy for assessing the sensitivity of the solution to the noise $R_{2,-}$ and the model uncertainty $\Delta \mathbf{A_1}$. Consequently, one may solve the optimisation problem \eqref{eq:minNormOptimization} over all data subsequences that satisfy the rank condition \eqref{eq:rankX-U}, and select the subsequence that yields the smallest value of $\|\hat{G}_\Theta\|_{F}$.
\end{remark}

The bound in \eqref{eq:boundPi}  relates the error in $\Delta \mathbf{\Theta}$  to the uncertainties introduced in the earlier stages, \ie, $d_2$, $\Delta x_2$ and $\Delta u_2$ 
via $R_{2,-}$, and $\Delta \mathbf{A_1}$. 
As we will demonstrate in the next section, this result enables a direct translation of the final uncertainty specifications of each application into corresponding requirements on the uncertainties at the source.

\section{Cascade Stabilisation}  \label{sec:cascadeStabilisation}
In this section, we apply the framework developed in Section \ref{sec:dataSolutionSyl} to address the data-driven cascade stabilisation problem via the forwarding method. After recalling the model-based method, we provide a data-driven formulation in both the noise-free and the noisy scenarios. 
Finally, we develop the data-driven forwarding method recursively to stabilise cascades containing an arbitrary number of subsystems.

\subsection{Problem Formulation} \label{sec:cascade-problem-formulation}
We consider the cascade system
\begin{subequations} \label{eq:cascade}
    \begin{align}
    &x_1(k+1) = A_1x_1(k) + B_1u_1(k),\label{eq:cascade-x}\\ 
    &x_2(k+1) = A_2x_2(k) + B_2x_1(k) \label{eq:cascade-y},
    \end{align}
\end{subequations} 
where $x_1(k) \in \mathbb{R}^{n_1}$ and $x_2(k) \in \mathbb{R}^{n_2}$ are the states of the first and second subsystems, respectively, and $u_1(k) \in \mathbb{R}^{m_1}$ is the input to the first subsystem (and subsequently to the entire cascade), with $A_1 \in \mathbb{R}^{n_1 \times n_1}$, $B_1 \in \mathbb{R}^{n_1 \times m_1}$, $A_2 \in \mathbb{R}^{n_2 \times n_2}$, and $B_2 \in \mathbb{R}^{n_2 \times n_1}$. We note that this system mirrors the interconnection of \eqref{Eq:SigmaSigma1} and \eqref{Eq:SigmaSigma2}, where $u_2(k) = y_1(k) = x_1(k)$.

The (model-based) stabilisation of such a system is achievable via the forwarding method, which involves finding a pre-stabilising feedback control law $u_1(k) = N_1 x_1(k) + v(k)$, where $N_1 \in \mathbb{R}^{m_1\times n_1}$ is such that $A_1 + B_1 N_1$ is Schur, and $v(k) \in \mathbb{R}^{m_1}$ is a new control input to be designed. As recalled in Section~\ref{sec:preliminaries-sylvester}, if $v(k) \equiv 0$ then the cascade interconnection defines an invariant subspace $\{(x_1,x_2) : x_2=\Upsilon_c x_1\}$ where $\Upsilon_c \in \mathbb{R}^{n_2 \times n_1}$ is the unique solution to the Sylvester equation
\begin{equation} 
    \label{eq:sylvester_cascade}
    \Upsilon_c(A_1 + B_1 N_1) = A_2\Upsilon_c + B_2,
\end{equation}
provided that $\lambda(A_1 + B_1 N_1) \cap \lambda(A_2) = \emptyset$. The above Sylvester equation takes the form \eqref{eq:SylEqGeneric}, with $\mathbf{\Theta} = \Upsilon_c$, $\mathbf{A_1} = (A_1 + B_1 N_1)$, $ \mathbf{A_2} = A_2$, $\mathbf{B_2} = B_2$ and $\mathbf{C_1}=I$. We then apply a change of coordinates $\zeta := x_2 - \Upsilon_c x_1$ where the transformed state $\zeta(k) \in \mathbb{R}^{n_2}$ is the error between $x_2$ and the defined invariant subspace, with dynamics
\begin{equation} \label{eq:zeta-system}
    \zeta(k+1) = A_2 \zeta(k) - \Upsilon_c B_1 v(k).
\end{equation}

In \eg{} \cite[Theorem 1.b]{simpson-porco-arxiv}, it is shown that\footnote{Note that we do not need to make these assumptions as the satisfaction of LMIs such as \eqref{eq:sdp-for-zeta-system} implies that $(A_2, -\Upsilon_c B_1)$ is stabilisable. This is the case for all analogous LMIs appearing in the data-driven results that we present.} if
$\lambda(A_1 + B_1 N_1) \cap \lambda(A_2) = \emptyset$, $(A_2, B_2)$ is stabilisable, and no eigenvalue of $A_2$ is a zero of the system $(A_1+B_1N_1,B_2,I)$ (this is known as \textit{non-resonance condition}), then
$(A_2, -\Upsilon_c B_1)$ is stabilisable. Thus, the feedback control law $v(k) := N_2 \zeta(k)$ can be used to stabilise \eqref{eq:zeta-system}, where $N_2 \in \mathbb{R}^{m_1\times n_2}$ is selected to ensure that $A_2 - \Upsilon_c B_1 N_2$ is Schur. Then, the cascade \eqref{eq:cascade} in closed-loop with the overall control law 
 $u_1(k) = N_1 x_1(k) + N_2 \zeta(k)$, 
results in the dynamics
\begin{equation} \label{eq:closed-loop-cascade-2-systems}
    \begin{bmatrix}
        x_1(k+1)\\\zeta(k+1)
    \end{bmatrix} = \begin{bmatrix}
        A_1 + B_1 N_1 & B_1 N_2\\ \mathbf{0} & A_2 - \Upsilon_c B_1 N_2
    \end{bmatrix}\begin{bmatrix}
        x_1(k)\\\zeta(k)
    \end{bmatrix}.
\end{equation}
We note the upper-triangular form of the dynamics matrix, indicating that the stability of each individual subsystem implies the stability of the whole cascade system. In what follows we address the problem of developing a direct data-driven version of this method. 

\subsection{The Noise-free Problem} 
\label{sec:cascade-full-info}
Consider the cascade system \eqref{eq:cascade}, and suppose now that the matrices $A_1, B_1, A_2, $ and $B_2$ are unknown, but the input data and state data for both subsystems are available so that the data matrices $U_{1,-}$, $X_{1,-}, X_{1,+}, X_{2,-},$ and $X_{2,+}$ can be constructed. We now propose a data-driven forwarding method to directly stabilise the unknown cascade system. We begin by introducing the following assumption on the data matrices.

\begin{assumption} \label{ass:rank-assumption}
The available data matrices are such that the rank conditions 
\begin{subequations} \label{eq:rank_sys}
\begin{align}
    &\rank\left(\begin{bmatrix}X_{1,-}\\U_{1,-} \label{eq:rank_sys_1}
    \end{bmatrix}\right) = n_1 + m_1,\\ 
    &\rank\left(\begin{bmatrix}X_{2,-}\\X_{1,-} \label{eq:rank_sys_2}
    \end{bmatrix}\right) = n_2 + n_1,
\end{align}
\end{subequations}
are satisfied.
\end{assumption}

Mirroring the model-based method outlined in Section \ref{sec:cascade-problem-formulation}, the data-driven forwarding method has three distinct steps:
\begin{enumerate}
    \item Pre-stabilise the first subsystem \eqref{eq:cascade-x}.
    \item Find $\Upsilon_c$, \textit{i.e.} the solution of the Sylvester equation \eqref{eq:sylvester_cascade}.
    \item Apply a coordinate transformation to the second subsystem \eqref{eq:cascade-y} and stabilise the transformed system \eqref{eq:zeta-system}.
\end{enumerate}
The pre-stabilisation Step 1 is analogous to the data-driven stabilisation methodology recalled 
in Section \ref{sec:preliminaries-data-driven}, whereby the stabilising gain is found through the solution of 
a data-dependent LMI of the type \eqref{eq:basic-sdp}. 
Then, the novel contributions of this section are data-driven methods for Steps 2 and 3, namely finding the solution of the underlying Sylvester equation, such that the invariant subspace $\{(x_1,x_2) : x_2=\Upsilon_c x_1\}$ can be characterised and Step 3 executed. 

\begin{lemma} \label{lem:cascade-data-sylvester-lemma}
Suppose that Assumption~\ref{ass:rank-assumption} holds and let $N_1 =U_{1,-}Q_1(X_{1,-}Q_1)^{-1}$ be such that $\lambda(A_1 + B_1 N_1) \cap \lambda(A_2) = \emptyset$, where $Q_1$ is any matrix that satisfies \eqref{eq:basic-sdp}, with $Q_{1} := Q_{x}$.
Then, any matrix $G_{\Upsilon_c} \in \mathbb{R}^{T\times n_1}$ that satisfies 
\begin{subequations} \label{eq:cascade-sylvester-sdp}
\begin{empheq}[left=\empheqlbrace]{align}
    X_{2,+}  G_{\Upsilon_c} &= X_{2,-} G_{\Upsilon_c} X_{1,+} G_{N_1}   \label{eq:dataSylvCascade1} \\
    X_{1, -}  G_{\Upsilon_c} &= I    \label{eq:dataSylvCascade2}
  \end{empheq}
\end{subequations}
where $G_{N_1}$ is obtained from \eqref{eq:basic-sdp-GK} with $K = N_1$, is such that 
\begin{equation} \label{eq:computing-sylv-from-sdp}
    \Upsilon_c = X_{2,-}G_{\Upsilon_c}.
\end{equation}
Namely \eqref{eq:computing-sylv-from-sdp} gives the unique solution of \eqref{eq:sylvester_cascade}. 
\end{lemma}

\begin{proof}
The result follows by applying Lemma~\ref{lemma:computePI} to \eqref{eq:sylvester_cascade}, with $\mathbf{A_1} = A_1+B_1N_1 = X_{1,+}G_{N_1}$ by \eqref{eq:linear-closed-loop-x1g}, $\mathbf{A_2} = A_2$, $\mathbf{B_2} = B_2$, $\mathbf{C_1} = I$ and $\mathbf{\Theta} = \Upsilon_c$, and 
noting that $U_{2,-} = X_{1,-}$ in the cascade, hence \eqref{eq:rank_sys_2} ensures that the conditions of Lemma~\ref{lemma:computePI} are satisfied.
\end{proof}
Having determined $\Upsilon_c$,  
we are ready to complete Step 3 of the data-driven forwarding method. Namely, we construct a feedback loop using the transformed state $\zeta(k)$  
to stabilise the cascade.

\begin{theorem} \label{thm:modified-2-system-casc}
Suppose that Assumption \ref{ass:rank-assumption} holds.
Let $N_1 := U_{1,-}Q_1(X_{1,-}Q_1)^{-1}$ be such that $\lambda(A_1 + B_1 N_1) \cap \lambda(A_2) = \emptyset$, where $Q_1$ is any matrix that satisfies \eqref{eq:basic-sdp}, with $Q_1 := Q_x$. Let $\Upsilon_c = X_{2,-}G_{\Upsilon_c}$, with $G_{\Upsilon_c}$ any matrix that satisfies \eqref{eq:cascade-sylvester-sdp}, be such that
\begin{equation}
    \rank\left(\begin{bmatrix}Z_{-}\\V_{-} \label{eq:rank_sys_3}
    \end{bmatrix}\right) = n_2 + m_1,
\end{equation}
with $Z_{-} := X_{2,-} - \Upsilon_c X_{1,-}$ and $V_{-} := U_{1,-} - N_1 X_{1,-}$, holds. Then the feedback control law
\begin{equation} \label{eq:casc-feedback-2-systems}
        u_{\text{casc}}(k) = N_1x_1(k) + N_2(x_2(k) - \Upsilon_c x_1(k)),
\end{equation}
with $N_2 := V_{-}Q_2(Z_{-}Q_2)^{-1}$, where $Q_2$ is any matrix that satisfies 
\begin{equation}
        \label{eq:sdp-for-zeta-system}
    \begin{bmatrix} 
        Z_{-}Q_2 & Z_{+} Q_2 \\ Q_2^\top Z_{+}^\top & Z_{-}Q_2
    \end{bmatrix} \succ 0,
\end{equation}
with $Z_{+} := X_{2,+} - \Upsilon_c X_{1,+}$, is such that the cascade system \eqref{eq:cascade} is asymptotically stable.
\end{theorem}

\begin{proof}
The data-driven design of the stabilising feedback \eqref{eq:casc-feedback-2-systems} mirrors the steps outlined in Section \ref{sec:cascade-problem-formulation}, whereby the pre-stabilisation of the $x_1$-dynamics of \eqref{eq:cascade} via $u_1(k) = N_1x_1(k)$ is a pre-requisite for the remaining steps of the forwarding method.
The rank condition \eqref{eq:rank_sys_1} and the formula of $N_1$ given in the statement achieves precisely this, as explained in Section \ref{sec:preliminaries-data-driven} (see also \cite[Theorem 3]{de2019formulas}). Then, by the rank condition \eqref{eq:rank_sys_2} we can apply Lemma~\ref{lem:cascade-data-sylvester-lemma} to determine $\Upsilon_c$, \ie, the solution of the Sylvester equation that characterises the invariant subspace $\{(x_2,x_1): x_2=\Upsilon_c x_1 \}$ that the second subsystem is desired to track. What is left to do is to design the gain $N_2$ that stabilises the system associated to the error $\zeta = x_2 - \Upsilon_c x_1$. To this end, recall that $X_{1,-}$, $X_{2,-}$, and $X_{2,+}$ satisfy $X_{2,+} = A_2X_{2,-} + B_2X_{1,-}$ by definition. Adding and subtracting $\Upsilon_c X_{1,+}$ and $A_2\Upsilon_c X_{1,-}$ to the previous equation yields
\begin{equation} \label{eq:transformed-zeta-linear}
    Z_{+} + \Upsilon_c X_{1,+} = A_2 Z_{-} + A_2\Upsilon_c X_{1,-} + B_2X_{1,-}.
\end{equation}
The matrices $Z_{-}$ and $Z_{+}$ can be interpreted as the data generated by the error $\zeta = x_2 - \Upsilon_c x_1$, which can be obtained directly from the already available data of $x_1$ and $x_2$. Recall also that $X_{1,+} = A_1 X_{1,-} + B_1 U_{1,-}$. Adding and subtracting $B_1N_1X_{1,-}$ in this equation yields
\begin{equation} \label{eq:transformed-x-linear}
    X_{1,+} = (A_1 + B_1N_1)X_{1,-} + B_1V_{-},
\end{equation}
where $V_{-} = U_{1,-} - N_1X_{1,-}$. Hence, substituting \eqref{eq:transformed-x-linear} into \eqref{eq:transformed-zeta-linear}, 
yields 
\begin{equation}
    \begin{split}
    Z_{+} + \Upsilon_c(A_1+B_1N_1)X_{1,-} + \Upsilon_c B_1V_{-} \\= A_2 Z_{-} + A_2\Upsilon_cX_{1,-} + B_2X_{1,-},
    \end{split}
\end{equation}
from which, using the Sylvester equation \eqref{eq:sylvester_cascade} to cancel out terms, we obtain $Z_{+} = A_2 Z_{-} - \Upsilon_c B_1 V_{-}$. At this point we can use the results recalled in Section \ref{sec:preliminaries-data-driven} to write
$
    \begin{bmatrix}
         A_2 & -\Upsilon_c B_1
    \end{bmatrix}= Z_{+} \begin{bmatrix} Z_{-}\\V_{-} \end{bmatrix}^{\dagger},
$
which is well defined by the condition \eqref{eq:rank_sys_3}. Choosing $v(k) = N_2\zeta(k)$ allows $N_2$ to be computed as $N_2 = V_{-}Q_2(Z_{-} Q_2)^{-1}$ where $Q_2$ is a solution of the LMI \eqref{eq:sdp-for-zeta-system}, see \cite[Theorem 3]{de2019formulas}. In conclusion, the stability of the cascade follows from the fact that the obtained $u_{\text{casc}}$ is the feedback of the forwarding method, where $N_1$, $\Upsilon_c$, and $N_2$ are constructed from data such that $A_1+B_1N_1$ and $A_2-\Upsilon_c B_1N_2$ are Schur. From \eqref{eq:closed-loop-cascade-2-systems}, this implies that the entire cascade is stabilised by $u_{\text{casc}}$.
\end{proof}

\begin{remark} \label{remark:rank-condition-cascade-2}
The rank conditions \eqref{eq:rank_sys_1}, \eqref{eq:rank_sys_2}, and \eqref{eq:rank_sys_3} imply a lower bound on the number of data points (\ie, $T$) required to stabilise the cascade via the proposed approach, namely
\begin{equation}
    T \ge \max \left(n_2 + m_1, n_1 + m_1, n_2 + n_1\right).
\end{equation}
Suppose instead that the cascade structure \eqref{eq:cascade} had been stabilised by being treated as a monolithic system with the aggregated state $\mathcal{X}(k):= \col(x_1(k),x_2(k)) \in \mathbb{R}^{n_1 + n_2}$ and dynamics
\begin{equation}
    \mathcal{X}(k+1) = \begin{bmatrix}
        A_1 & \mathbf{0} \\B_2 & A_2
    \end{bmatrix} \mathcal{X}(k) + \begin{bmatrix}
        B_1\\ \mathbf{0}
    \end{bmatrix}u_1(k).
\end{equation}
The data-driven stabilisation of such a system (following the procedure recalled in Section \ref{sec:preliminaries-data-driven}) would require the minimal $T$ to be
\begin{equation}
    T \ge n_2 + n_1 + m_1.
\end{equation}
This observation reveals a benefit of the proposed data-driven forwarding method, regarding data efficiency. Namely, by exploiting the cascade structure of \eqref{eq:cascade}, fewer data points are required in order to stabilise the overall system. While this advantage in data efficiency may be marginal when only two subsystems are involved, we highlight that this becomes increasingly significant when the number of cascade stages grows, as will be shown in Section \ref{sec:N-cascade}. The caveat here is that during the data generation experiment, one only has direct influence on the input data $U_{1,-}$. This means that the rank condition \eqref{eq:rank_sys_1} can be satisfied easily, but it is not guaranteed that this will lead to the subsequent rank conditions \eqref{eq:rank_sys_2}, \eqref{eq:rank_sys_3}, and all those resulting from additional cascaded systems 
being satisfied. 
\end{remark}

\subsection{Cascade Stabilisation with Noise-corrupted Data} \label{sec:cascade-noise}
Consider now a noisy version of \eqref{eq:cascade} given by
\begin{subequations} \label{eq:cascadeNoise}
    \begin{align}
    &x_1(k+1) = A_1x_1(k) + B_1u_1(k) + d_1(k),\label{eq:cascade-x-noise}\\ 
    &x_2(k+1) = A_2x_2(k) + B_2x_1(k) + d_2(k) \label{eq:cascade-y-noise},
    \end{align}
\end{subequations} 
where $d_1(k)$ and $d_2(k)$ represent additive system noise. 
Suppose further that only corrupted measurements of the states $x_1(k)$ and $x_2(k)$, given by
\begin{subequations} \label{eq:cascade-measurements}
\begin{align}
    \bar{x}_1(k) &= x_1(k) + \Delta x_1(k),\\ 
    \bar{x}_2(k) &= x_2(k) + \Delta x_2(k),
\end{align}
\end{subequations}
with $\Delta x_1(k)$ and $\Delta x_2(k)$ denoting the measurement noise, 
are available. In the presence of these unknown disturbances, it is natural to shift the focus from the internal Lyapunov stability studied in the noise-free case to characterising how such perturbations affect the closed-loop behaviour.
Accordingly, the aim of this subsection is to design a controller of the form
\begin{equation} \label{eq:controllerNoisyCascade}
    u_1(k) = N_1 \bar{x}_1(k) + N_2 (\bar{x}_2(k) - \hat{\Upsilon}_c \bar{x}_1(k)),
\end{equation}
where the matrices $N_1, N_2$ and $\hat{\Upsilon}_c$ are to be synthesised from (noisy) data, such that the resulting closed-loop dynamics of system \eqref{eq:cascadeNoise}-\eqref{eq:controllerNoisyCascade} possess certain robustness properties with respect to these perturbations. 


Following the treatment in Section \ref{sec:noiseErrorSyl}, we define the encapsulated noise signals for both subsystems of the cascade as
\begin{subequations} \label{eq:cascade-noise-R}
\begin{align}
    r_1(k) &:= A_1 \Delta x_1(k) - \Delta x_1(k+1) - d_1(k) \label{eq:cascade-noise-r-x1},\\ 
    r_2(k) &:= A_2 \Delta x_2(k) - \Delta x_2(k+1) + B_2 \Delta x_1(k) -d_2(k)\label{eq:cascade-noise-r-x2}.
\end{align}
\end{subequations}
Such definitions provide the following compact description of the dynamics of the corrupted measurements of system \eqref{eq:cascadeNoise}
\begin{subequations} \label{eq:cascadeMeasuremntSystem}
\begin{align}
    \bar{x}_1(k+1) &= A_1  \bar{x}_1(k) + B_1 u_1(k) -r_1(k), \label{eq:cascadeMeasuremntSystem1} \\
    \bar{x}_2(k+1) &= A_2  \bar{x}_2(k) + B_2 \bar{x}_1(k) - r_2(k).
\end{align}
\end{subequations}
The above ``measurement dynamics'' are useful, as they allow for most of the analysis and design in this subsection to be directly conducted in the \textit{measurement space}, where the control law \eqref{eq:controllerNoisyCascade} also lies.
We first 
provide a ``noisy version'' of Assumption~\ref{ass:rank-assumption}.

\begin{assumption} \label{ass:noisy-cascade-assumption}
The available data matrices are such that the rank conditions
\begin{subequations} \label{eq:rank-noise-casc-1}
\begin{align}    
&\rank\left(\begin{bmatrix}\,\bar{X}_{1,-}\\ U_{1,-} 
    \end{bmatrix}\right) = n_1 + m_1, \label{eq:rank_sys_noise_1} \\ 
    &\rank\left(\begin{bmatrix}\,\bar{X}_{2,-}\\\bar{X}_{1,-} 
    \end{bmatrix}\right) = n_2 + n_1, \label{eq:rank_sys_noise_2}
\end{align}
\end{subequations}
are satisfied. 
\end{assumption}

As usual, the first design step focuses on determining the pre-stabilisation gain $N_1$ (from the corrupted measurements) such that $A_1 + B_1 N_1$ is Schur. This is a direct application of the noise-robust result recalled in Section~\ref{sec:preliminaries-data-driven}, which for self-containedness we formulate in the following remark. 

\begin{remark} \label{remark:N1-noise-cs}
Suppose that the rank condition \eqref{eq:rank_sys_noise_1} holds. Suppose that there exists a scalar $\alpha_1 > 0$ such that 
\begin{equation} \label{eq:SNR-condtion-x1}
    R_{1,-} R_{1,-}^\top \preceq \frac{\alpha_1^2}{2(2+\alpha_1)} \bar{X}_{1,+} \bar{X}_{1,+}^\top.
\end{equation}
Then \cite{de2019formulas}, the gain $N_1 := U_{1,-}Q_{x}(\bar{X}_{1,-}Q_{x})^{-1}$, where $Q_{x} \in \mathbb{R}^{T\times n_1}$ is any matrix satisfying
\begin{subequations}
    \label{eq:sdp-for-noisy-system-N1}
    \begin{align}
    &\begin{bmatrix} 
        \bar{X}_{1,-}Q_x - \alpha_1 \bar{X}_{1,+} \bar{X}_{1,+}^\top& \bar{X}_{1,+}Q_x \\ Q_x^\top \bar{X}_{1,+}^\top & \bar{X}_{1,-}Q_x
    \end{bmatrix} \succ 0, \\
    &\begin{bmatrix} 
        I & Q_x \\ Q_x^\top & \bar{X}_{1,-} Q_x
    \end{bmatrix} \succ 0,
     \end{align}
\end{subequations}
renders the matrix $A_1 + B_1 N_1$ Schur. Moreover, $A_1 + B_1 N_1 = (\bar{X}_{1,+} + R_{1, -}) G_{N_1}$, with $G_{N_1}$ any matrix satisfying
\begin{equation*}
    \begin{bmatrix}
        I \\ N_1
    \end{bmatrix}=\begin{bmatrix}
        \bar{X}_{1,-} \\ U_{1,-}
    \end{bmatrix} G_{N_1}.
\end{equation*}
\end{remark}

In the noise-free case, the second design step consists of finding the solution $\Upsilon_c$ to the Sylvester equation \eqref{eq:sylvester_cascade}, which defines the coordinate transformation useful for the final post-stabilisation step, through solving the LP \eqref{eq:cascade-sylvester-sdp}. However, in the noisy setting considered in this subsection, one can only obtain an empirical solution $\hat{\Upsilon}_c :=  \bar{X}_{2,-} \hat{G}_{\Upsilon_c}$, with $\hat{G}_{\Upsilon_c} \in \mathbb{R}^{T\times n_1}$ a solution of 
\begin{subequations} \label{eq:empirical-LMI-Upsilon-noise}
\begin{empheq}[left=\empheqlbrace]{align}
    \bar{X}_{2,+}  \hat{G}_{\Upsilon_c} &= \bar{X}_{2,-} \hat{G}_{\Upsilon_c} (\bar{X}_{1,+} G_{N_1}),   \\
    \bar{X}_{1, -}  \hat{G}_{\Upsilon_c} &= I,    
  \end{empheq}
\end{subequations}
where $\bar{X}_{1,+} G_{N_1}$ serves as an empirical approximant of $A_1 + B_1 N_1$, in line with the discussions in Section \ref{sec:noiseErrorSyl}. In what follows, we provide some instrumental results on the empirical error of the solution, leveraging Lemma~\ref{lemma:errorSyl} and Theorem~\ref{lemma:errorBound}.

\begin{lemma} \label{lemma:errorBound-cascade}
Suppose that Assumption \ref{ass:noisy-cascade-assumption} holds and that $\lambda(A_1 +B_1 N_1) \cap \lambda(A_2) = \emptyset$. Define $\Delta \Upsilon_c := \Upsilon_c  - \hat{\Upsilon}_c$.
Then, the following results hold.
\begin{enumerate}
    \item[(i)] $\Delta \Upsilon_c$ is the unique solution of the Sylvester equation
\begin{equation} \label{err-sylvester-cs}
\begin{aligned}
    A_2 \Delta \Upsilon_c & -  \Delta \Upsilon_c (A_1 +B_1 N_1) \\
    &= - R_{2, -} \hat{G}_{\Upsilon_c} + \bar{X}_{2,-} \hat{G}_{\Upsilon_c} R_{1, -} G_{N_1}. 
\end{aligned}
\end{equation}
    \item[(ii)] The bound
\begin{equation*} 
\begin{split} 
    \|\Delta & \Upsilon_c \|_2 \\
    &\leq \frac{\|\hat{G}_{\Upsilon_c}\|_F \left(\| R_{2, -} \|_2 + \|\bar{X}_{2, -}\|_2 \|R_{1,-}\|_2 \|G_{N_1}\|_2 \right)}{\sigma_{min}(I \otimes A_2 - (A_1 +B_1 N_1)^\top \otimes I)} 
\end{split}
\end{equation*}
holds. 
\end{enumerate}
\end{lemma}

\begin{proof}
Noting that $\Delta \mathbf{A}_1 = - R_{1, -} G_{N_1}$, $\mathbf{A_1} = A_1 +B_1 N_1$, $\mathbf{A_2} = A_2$ and $\hat{\mathbf{\Theta}} = \hat{\Upsilon}_c$, the Sylvester equation follows directly 
from Lemma \ref{lemma:errorSyl} and the claimed bound follows from Theorem \ref{lemma:errorBound} by further recognising that $\|R_{1,-} G_{N_1}\|_2 \leq \|R_{1,-}\|_2 \|G_{N_1}\|_2$ by the sub-multiplicative property.
\end{proof}

Having obtained the matrix $\hat{\Upsilon}_c$ that determines the (empirical) coordinate transformation, we define the nominal and the empirical transformed states (in the measurement space). 
\begin{equation}\label{eq:overzetahatzeta}
    \bar{\zeta}(k) := \bar{x}_2(k) - \Upsilon_c \bar{x}_1(k),  \quad
    \hat{\zeta}(k) := \bar{x}_2(k) - \hat{\Upsilon}_c \bar{x}_1(k),
\end{equation}
respectively, and note that $\hat{\zeta}(k) = \bar{\zeta}(k) + \Delta \Upsilon_c \bar{x}_1(k)$. Let $u_1(k) = N_1 \bar{x}_1(k) + v(k)$ and note that the dynamics of $\bar{\zeta}(k)$ satisfies
\begin{equation} \label{eq:zeta-bar-system-cascade}
\begin{aligned}
    \bar{\zeta}(k+1) = &  \bar{x}_2(k+1) - \Upsilon_c \bar{x}_{1}(k+1) \\
    \overset{\eqref{eq:cascadeMeasuremntSystem}}{=} &  A_2  \bar{x}_2(k) + B_2 \bar{x}_1(k) - r_2(k) \\
    & - \Upsilon_c (A_1  \bar{x}_1(k) + B_1 u_1(k) -r_1(k)) \\
    = &  A_2  \bar{x}_2(k) + B_2 \bar{x}_1(k) - r_2(k) \\
    & - \Upsilon_c ((A_1 + B_1 N_1)  \bar{x}_1(k) + B_1 v(k) -r_1(k)) \\
    = &  A_2  \bar{x}_2(k) + (B_2 - \Upsilon_c(A_1 + B_1 N_1)) \bar{x}_1(k) - r_2(k) \\
    & - \Upsilon_c (B_1 v(k) - r_1(k)) \\
    \overset{\eqref{eq:sylvester_cascade}}{=} &  A_2  (\bar{x}_2(k) - \Upsilon_c \bar{x}_1(k)) - \Upsilon_c B_1 v(k)  \\
    & + \Upsilon_c r_1(k) - r_2(k) \\
    = &  A_2  \bar{\zeta}(k) - \Upsilon_c B_1 v(k) 
    + \Upsilon_c r_1(k) - r_2(k) \\
    = &  A_2  \bar{\zeta}(k) - \hat{\Upsilon}_c B_1 v(k) \\
    &+ \Upsilon_c r_1(k) - r_2(k) - \Delta \Upsilon_c B_1 v(k) ,
\end{aligned}
\end{equation}
where 
in the last equality we used the definition of $\Delta \Upsilon_c$. Thus, we note that the term $\Upsilon_c r_1(k) - r_2(k) - \Delta \Upsilon_c B_1 v(k)$ plays the role of the ``process noise'' of $\bar{\zeta}(k)$. Together with the observation from \eqref{eq:overzetahatzeta} that $\Delta \Upsilon_c \bar{x}_1(k) = \hat{\zeta}(k) - \bar{\zeta}(k)$ plays the role of the ``measurement noise'' of $\bar{\zeta}(k)$,  
it becomes clear that designing $N_2$ such that $v(k) = N_2 \hat{\zeta}(k)$ renders $A_2 -\hat{\Upsilon}_c B_1 N_2$ Schur is another iteration of the noise-robust result in Section~\ref{sec:preliminaries-data-driven}, with the SNR condition derived in the next result.

\begin{lemma} \label{lemma:N2-noise-cs}
Consider the $\bar{\zeta}$-dynamics \eqref{eq:zeta-bar-system-cascade}. Suppose that the rank condition
\begin{equation}
    \rank\left(\begin{bmatrix}
        \hat{Z}_{-}\\ V_{-}
    \end{bmatrix}\right) = n_2 + m_1,   
\end{equation}
with $\hat{Z}_{-} := \bar{X}_{2, -} - \hat{\Upsilon}_c \bar{X}_{1, -}$ and $V_{-} := U_{1, -} - N_1 \bar{X}_{1, -}$, holds.
Suppose there exists a scalar $\alpha_2 > 0$ such that
\begin{equation} \label{eq:SNR-condtion-zeta}
    R_{\zeta,-}R_{\zeta,-}^\top \preceq \frac{\alpha_2^2}{2(2+\alpha_2)}\hat{Z}_{+}\hat{Z}_{+}^\top,
\end{equation}
where $\hat{Z}_{+} := \bar{X}_{2, +} - \hat{\Upsilon}_c \bar{X}_{1, +}$ and
\begin{equation} \label{eq:what-is-r-sigma}
\begin{split}
    R_{\zeta,-} &:= R_{2,-} -  \bar{X}_{2,-} \hat{G}_{\Upsilon_c} R_{1,-} - R_{2,-}\hat{G}_{\Upsilon_c}\bar{X}_{1,-} \\
    &+ \bar{X}_{2,-} \hat{G}_{\Upsilon_c} R_{1,-} G_{N_1} \bar{X}_{1,-}.
\end{split}
\end{equation}
Then the gain
$N_2 := V_{-}Q_{\zeta}(\hat{Z}_{-} Q_{\zeta})^{-1}$, where $Q_{\zeta} \in \mathbb{R}^{T\times n_2}$ is any matrix satisfying
\begin{subequations}
    \label{eq:sdp-for-noisy-system-N2}
    \begin{align}
    &\begin{bmatrix} 
        \hat{Z}_{-} Q_{\zeta} - \alpha_2 \hat{Z}_{+} \hat{Z}_{+}^\top & \hat{Z}_{+} Q_{\zeta} \\ Q_{\zeta}^\top \hat{Z}_{+}^\top & \hat{Z}_{-} Q_{\zeta}
    \end{bmatrix} \succ 0, \\
    &\begin{bmatrix} 
        I & Q_{\zeta} \\ Q_{\zeta}^\top & \hat{Z}_{-} Q_x
    \end{bmatrix} \succ 0,
     \end{align}
\end{subequations}
renders the matrix $A_2 -\hat{\Upsilon}_c B_1 N_2$ Schur.
\end{lemma}

\begin{proof}
The construction of $R_{\zeta,-}$ follows the same idea used to construct $r_1(k)$ and $r_2(k)$ as in \eqref{eq:cascade-noise-R}, which is to find an encapsulated signal $r_\zeta(k)$ such that the $\hat{\zeta}$-dynamics can be expressed as
\begin{equation} \label{eq:zeta-hat-dynamics}
    \hat{\zeta}(k+1) =  A_2  \hat{\zeta}(k) - \hat{\Upsilon}_c B_1 v(k) - r_\zeta(k),
\end{equation}
and then the claim boils down to a direct consequence of the noise-robust result in Section \ref{sec:preliminaries-data-driven}. 
Thus, the proof consists in deriving \eqref{eq:what-is-r-sigma}. To this end, recall that $ \hat{\zeta}(k)  =  \bar{\zeta}(k) + \Delta \Upsilon_c \bar{x}_1(k)$. By writing this for $k+1$ and using the $\bar{\zeta}$-dynamics \eqref{eq:zeta-bar-system-cascade}, we obtain \eqref{eq:zeta-hat-dynamics} where $r_\zeta(k)$ is given by
\begin{equation*}
\begin{aligned}
    r_\zeta(k) := & A_2 \Delta \Upsilon_c \bar{x}_1(k) -  \Delta \Upsilon_c \bar{x}_1(k+1) \\
    &- (\Upsilon_c r_1(k) - r_2(k) - \Delta \Upsilon_c B_1 v(k)) \\
    \overset{\eqref{eq:cascadeMeasuremntSystem1}}{=}& A_2 \Delta \Upsilon_c \bar{x}_1(k) -  \Delta \Upsilon_c (A_1  \bar{x}_1(k) + B_1 u_1(k) -r_1(k)) \\
    &- (\Upsilon_c r_1(k) - r_2(k) - \Delta \Upsilon_c B_1 v(k)) \\
    =& A_2 \Delta \Upsilon_c \bar{x}_1(k) -  \Delta \Upsilon_c (A_1 + B_1 N_1)  \bar{x}_1(k) \\
    &+ \Delta \Upsilon_c r_1(k) - \Upsilon_c r_1(k) + r_2(k) \\
    \overset{\eqref{err-sylvester-cs}}{=}& (- R_{2, -} \hat{G}_{\Upsilon_c} + \bar{X}_{2,-} \hat{G}_{\Upsilon_c} R_{1, -} G_{N_1})  \bar{x}_1(k) \\
    &- \hat{\Upsilon}_c r_1(k) + r_2(k), \\
\end{aligned}
\end{equation*}
where we used 
$u_1(k) = N_1 \bar{x}_1(k) + v(k)$ in the third equality. 
By recalling that $\hat{\Upsilon}_c = \bar{X}_{2,-} \hat{G}_{\Upsilon_c}$, the definition of the data matrix $R_{\zeta, -}$ in \eqref{eq:what-is-r-sigma} follows. This completes the proof. 
\end{proof}

\begin{remark}
It is noteworthy that as we went through the three steps of the forwarding method, we always characterised all the errors as explicit expressions of the source uncertainties $R_{1,-}$ and $R_{2,-}$. For the second step, which involved solving \eqref{eq:empirical-LMI-Upsilon-noise}, the error $\Delta \Upsilon_c$ is connected back to $R_{1,-}$ and $R_{2,-}$ via \eqref{err-sylvester-cs}. Likewise, for the third step, the uncertainty $R_{\zeta, -}$
in \eqref{eq:what-is-r-sigma} is expressed in terms of $R_{1,-}$ and $R_{2,-}$ (while the remaining terms therein can be evaluated empirically). 
\end{remark}

Before proceeding with the analysis of the closed-loop system trajectories, we first determine the closed-loop dynamics in the measurement space.

\begin{lemma}
\label{lemma-closeloop2cascade}
Consider system \eqref{eq:cascadeMeasuremntSystem}, the empirical transformed state $\hat{\zeta}$, and the control \eqref{eq:controllerNoisyCascade}. Let $\Delta \Upsilon_c$ be defined as in Lemma \ref{lemma:errorBound-cascade}. Then, the closed-loop dynamics of $\col(\bar{x}_1(k), \hat{\zeta}(k))$ are described by
\begin{equation}   \label{eq:noisy-overal-cl-system}
    \begin{bmatrix}
        \bar{x}_1(k+1)\\ \hat{\zeta}(k+1)
    \end{bmatrix} 
    = \mathcal{A}
    \begin{bmatrix}
        \bar{x}_1(k)\\ \hat{\zeta}(k)
    \end{bmatrix} 
    + \mathcal{B}
    \begin{bmatrix}
        r_1(k) \\r_2(k)
    \end{bmatrix},
\end{equation}
where
\begin{equation*}
    \mathcal{A} := \begin{bmatrix}
        A_1 + B_1 N_1 & B_1 N_2 \\ 
        \Delta \Upsilon_c (A_1 +B_1 N_1) - A_2 \Delta \Upsilon_c & A_2 - \hat{\Upsilon}_c B_1 N_2
    \end{bmatrix},
\end{equation*}
\begin{equation*}
    \mathcal{B} := \begin{bmatrix}
        -I & \mathbf{0} \\
        \hat{\Upsilon}_c & -I
    \end{bmatrix}.
\end{equation*}
\end{lemma}
\begin{proof}
From \eqref{eq:cascadeMeasuremntSystem1} and \eqref{eq:controllerNoisyCascade} it follows that
\begin{equation} \label{eq:proof_xbar_dynamics}
\bar{x}_1(k+1) = (A_1 + B_1 N_1) \bar{x}_1(k) + B_1 N_2 \hat{\zeta}(k) - r_1(k).
\end{equation}
For the variable $\hat{\zeta}(k)$, we have 
that
\begin{equation} \label{eq:proof-hatzeta-dynamics}
\begin{aligned}
\hat{\zeta}(k+1) & = \bar{x}_2(k+1) - \hat{\Upsilon}_c \bar{x}_1(k+1)  \\
&= A_2  \bar{x}_2(k) + B_2 \bar{x}_1(k) - r_2(k) \\
& \quad -\! \hat{\Upsilon}_c(A_1  \bar{x}_1(k) \!+\! B_1 (N_1 \bar{x}_1(k) \!+\! N_2 \hat{\zeta}(k)) \!-\!r_1(k)) \\
&= A_2  \bar{x}_2(k) + (B_2 - \hat{\Upsilon}_c(A_1 + B_1 N_1))\bar{x}_1(k) \\
& \quad - \hat{\Upsilon}_c B_1 N_2 \hat{\zeta}(k) - (r_2(k) - \hat{\Upsilon}_c r_1(k)).
\end{aligned}
\end{equation}
Substituting $\Upsilon_c = \hat{\Upsilon}_c + \Delta \Upsilon_c$ into the Sylvester equation \eqref{eq:sylvester_cascade}, we obtain that
\begin{equation} 
\begin{aligned}
    B_2 - & \hat{\Upsilon}_c(A_1 +  B_1 N_1)  \\
    &= - A_2 \Delta \Upsilon_c +  \Delta \Upsilon_c (A_1 +B_1 N_1) - A_2 \hat{\Upsilon}_c.
\end{aligned}
\end{equation}
Exploiting the above equation, \eqref{eq:proof-hatzeta-dynamics} becomes
\begin{equation} \label{eq:proof_zetahat_dynamics}
\begin{aligned}
\hat{\zeta}(k+1)  
&= (A_2 - \hat{\Upsilon}_c B_1 N_2) \hat{\zeta}(k) \\
& \quad + (- A_2 \Delta \Upsilon_c +  \Delta \Upsilon_c (A_1 +B_1 N_1))\bar{x}_1(k) \\
& \quad - (r_2(k) - \hat{\Upsilon}_c r_1(k)).
\end{aligned}
\end{equation}
Combining the dynamics in \eqref{eq:proof_xbar_dynamics} and \eqref{eq:proof_zetahat_dynamics} proves the claim.\end{proof}

Lemma~\ref{lemma-closeloop2cascade} shows that the presence of the process and measurement noise destroys the diagonal structure which was leveraged by the forwarding method to achieve (asymptotic) Lyapunov stability. Nevertheless, we can derive an input-to-state stability-like result.

\begin{theorem} \label{theorem:ISS-cascade} 
Consider system \eqref{eq:cascadeNoise} and the controller \eqref{eq:controllerNoisyCascade}. Suppose that the SNR conditions \eqref{eq:SNR-condtion-x1} and \eqref{eq:SNR-condtion-zeta} hold for some positive scalars $\alpha_1$ and $\alpha_2$. Let $N_1$, $N_2$, and $\hat{\Upsilon}_c$ be obtained using Remark \ref{remark:N1-noise-cs}, Lemma \ref{lemma:N2-noise-cs}, and \eqref{eq:empirical-LMI-Upsilon-noise}, respectively. Let the blocks of $\mathcal{A}$ as in Lemma~\ref{lemma-closeloop2cascade} be denoted
$ 
 [\mathcal{A}_{11}, \mathcal{A}_{12}; \mathcal{A}_{21}, \mathcal{A}_{22}].
$
Suppose that $R_{1,-}$ and $R_{2,-}$ are such that
\begin{equation} \label{eq:stability-margin-condition}
\begin{aligned}
    \| R_{2, -} & \hat{G}_{\Upsilon_c} - \bar{X}_{2,-} \hat{G}_{\Upsilon_c} R_{1, -} G_{N_1}\|_2 \\
    & < \frac{1}{\|\mathcal{A}_{12}\|_2 \|(zI-\mathcal{A}_{11})^{-1}\|_{\mathcal{H}_\infty} \|(zI-\mathcal{A}_{22})^{-1}\|_{\mathcal{H}_\infty}},
\end{aligned}
\end{equation}
where $\|\cdot\|_{\mathcal{H}_\infty} := \sup_{|z|=1}\sigma_{max}(\cdot)$. Define $\mathring{\zeta}(k) := x_2(k) - \hat{\Upsilon}_c x_1(k)$. Then, the trajectories of system \eqref{eq:cascadeNoise} in closed-loop with \eqref{eq:controllerNoisyCascade} satisfy
\begin{equation} \label{eq:cascade-noise-iss-bound}  
\begin{aligned}
    \left\|\begin{bmatrix}
    x_1(k) \\ \mathring{\zeta}(k)
    \end{bmatrix}\right\|_2 
    \le & \beta\left(
    \left\|\begin{bmatrix}
    \bar{x}_1(0)\\
    \hat{\zeta}(0)
    \end{bmatrix}\right\|_2, k\right) 
    + \gamma\left(\left|\begin{bmatrix}
            r_1\\r_2
        \end{bmatrix}\right|_{k-1}\right) \\
    &+ \left\|\begin{bmatrix}
        \Delta x_1(k) \\ \Delta x_1(k) - \hat{\Upsilon}_c \Delta x_2(k)
    \end{bmatrix}\right\|_2,
\end{aligned}
\end{equation}
where $|z|_k := \sup_{0\le j \le k}\|z(j)\|$, $\beta$ is a class-$\mathcal{K}\mathcal{L}$ function and $\gamma$ is a class-$\mathcal{K}$ function given as
\begin{equation*}
    \beta(s, k) := c p^k s, \quad \gamma(s) := \frac{c\| \mathcal{B} \|_2}{1-p} s,
\end{equation*}
with some constants $c > 0$ and $0 \leq p < 1$.
\end{theorem}

\begin{proof}
We first note that $\mathcal{A}_{11}$ and $\mathcal{A}_{22}$ are guaranteed to be Schur by the obtained $N_1$ and $N_2$, as established in Remark \ref{remark:N1-noise-cs} and Lemma \ref{lemma:N2-noise-cs}. 
This ensures that the terms $\|(zI-\mathcal{A}_{11})^{-1}\|_{\mathcal{H}_\infty}$ and $\|(zI-\mathcal{A}_{22})^{-1}\|_{\mathcal{H}_\infty}$ in \eqref{eq:stability-margin-condition} are finite. 
Then, by the result (i) in Lemma \ref{lemma:errorBound-cascade}, \eqref{eq:stability-margin-condition} yields that
$$
\|\mathcal{A}_{21}\|_2 \|\mathcal{A}_{12}\|_2 \|(zI-\mathcal{A}_{11})^{-1}\|_{\mathcal{H}_\infty} \|(zI-\mathcal{A}_{22})^{-1}\|_{\mathcal{H}_\infty} < 1.
$$
By exploiting that $\|\cdot\|_{\mathcal{H}_{\infty}} = \|\cdot\|_{2}$ for constant matrices and the sub-multiplicative property of the $\mathcal{H}_{\infty}$ induced operator norm, we have 
\begin{equation*} 
    \|(zI-\mathcal{A}_{11})^{-1} \mathcal{A}_{12}\|_{\mathcal{H}_\infty} \|(zI-\mathcal{A}_{22})^{-1} \mathcal{A}_{21}\|_{\mathcal{H}_\infty} < 1.
\end{equation*}
Note that this is a small-gain condition which ensures the asymptotic stability of system \eqref{eq:noisy-overal-cl-system} under $r_1 \equiv 0$ and $r_2 \equiv 0$, finally concluding that $\mathcal{A}$ is Schur. 
Then, the input-to-state stability of \eqref{eq:noisy-overal-cl-system} follows with the provided $\beta \in \mathcal{K}\mathcal{L}$
and $\gamma \in \mathcal{K}$ (see \cite[Example 3.4]{jiang2001input}), namely
\begin{equation} \label{eq:proof-iss-formula}
\begin{aligned}
    \left\|\begin{bmatrix}
    \bar{x}_1(k) \\ \hat{\zeta}(k)
    \end{bmatrix}\right\|_2 
    \le & \,\beta\left(
    \left\|\begin{bmatrix}
    \bar{x}_1(0)\\
    \hat{\zeta}(0)
    \end{bmatrix}\right\|_2, k\right) 
    + \gamma\left(\left|\begin{bmatrix}
            r_1 \\ r_2
        \end{bmatrix}\right|_{k-1}\right).
\end{aligned}
\end{equation}
Last, recognising that $x_1(k) = \bar{x}_1(k) - \Delta x_1(k)$ and $\mathring{\zeta}(k) = \hat{\zeta}(k) - (\Delta x_1(k) - \hat{\Upsilon}_c \Delta x_2(k))$ yields
\begin{equation*}
     \left\|\begin{bmatrix}
    x_1(k) \\ \mathring{\zeta}(k) 
    \end{bmatrix}\right\|_2 
    \leq
     \left\|\begin{bmatrix}
    \bar{x}_1(k) \\ \hat{\zeta}(k)
    \end{bmatrix}\right\|_2
    +
    \left\|\begin{bmatrix}
        \Delta x_1(k) \\ \Delta x_1(k) - \hat{\Upsilon}_c \Delta x_2(k)
    \end{bmatrix}\right\|_2, 
\end{equation*}
which proves the claim.
\end{proof}

\begin{remark}
Although the condition \eqref{eq:stability-margin-condition} on $R_{1,-}$ and $R_{2,-}$ cannot be verified using data (because the right-hand side requires $A_1$, $A_2$, and $B_1$), it provides the insight that one should keep the decision variable $\hat{G}_{\Upsilon_c}$ as small as possible, suggesting the formulation of a least-norm problem as already described in Remark~\ref{remark:minmize-norm-G}. Moreover, it also shows that if the noises $R_{2,-}$ and $R_{1,-}$ are sufficiently small, then \eqref{eq:stability-margin-condition} will be satisfied. 
\end{remark}

\subsection{Extending to $\mathcal{N}$ Subsystems}
\label{sec:N-cascade}
In this section, we demonstrate how the data-driven forwarding method proposed in Section \ref{sec:cascade-full-info} can be implemented in a recursive manner, to stabilise a cascade of a number of subsystems greater than $\mathcal{N} = 2$. Namely, consider
\begin{equation} \label{eq:n-system-cascade}
    x_i(k+1) = A_ix_i(k) +  \begin{cases} 
    B_1 u_1(k),& \text{if } i = 1\\
    B_ix_{i-1}(k), & \text{if } i \geq 2
    \end{cases}
\end{equation}
for $i \in [1,\mathcal{N}]$, where $x_i(k) \in \mathbb{R}^{n_i}$, $u_1(k) \in \mathbb{R}^{n_0}$, $A_i \in \mathbb{R}^{n_i \times n_i}$, and $B_i \in \mathbb{R}^{n_i \times n_{i-1}}$, and the
the overall control law
\begin{equation} \label{eq:control-n-cascade}
    u_1(k) 
    = \sum_{i=1}^{\mathcal{N}} N_i \zeta_i(k),
\end{equation}
where $N_i \in \RR^{n_0 \times n_i}$ are the feedback gains 
and $\zeta_i(k) \in \RR^{n_i}$ are the 
``transformed'' states (defined in a way to be specified shortly), for $i\in [1, \mathcal{N}]$. 
The general idea in recursively applying the forwarding method for stabilisation is then as follows: provided that the first $(j-1)$ subsystems (with $j \in [2, \mathcal{N}]$) are stabilised by
\begin{equation} \label{eq:control-j-cascade}
    u_1(k) = \sum_{i=1}^{j-1} N_i \zeta_i(k),
\end{equation}
we find a transformation matrix $\Upsilon_{c,j-1} \in \RR^{n_i \times \sum_{i=1}^{j-1} n_i}$ such that the subspace 
$\{(\zeta_{1:j-1},x_{j}): x_{j}=\Upsilon_{c,j-1}\zeta_{1:j-1}\}$, with 
$$
\zeta_{1:j-1} := \col(\zeta_1(k), \dots, \zeta_{j-1}(k)),
$$ 
is rendered invariant and attractive 
by adding $v_{j}(k) := N_j \zeta_j(k)$ to the cascade input $u_1(k)$ given in \eqref{eq:control-j-cascade}. Then, the stability of the cascade will be extended to include the $j$-th subsystem. Applying this recursively for all $\mathcal{N}$ subsystems yields the overall control law \eqref{eq:control-n-cascade}, which renders the overall cascade stable. 
To formalise this idea 
and specify all the involved objects, we provide below a model-based result that is instrumental to presenting the ultimate goal, \ie, a data-driven approach to stabilise the $\mathcal{N}$-cascade system. To streamline the presentation, for all $i \in [1,\mathcal{N}]$, let 
\begin{equation} \label{eq:what-is-Acl-i-model-based}
\begin{aligned}
    A_{cl,1:i} & := \\ 
    & \begin{cases}
    A_1+ B_1N_1, & \text{if } i = 1\\
    \begin{bmatrix}
        A_{cl,1:i-1} & B_{cl,1:i-1}N_i \\ \mathbf{0} & A_i - \Upsilon_{c,i-1}B_{cl,1:i-1}N_i
    \end{bmatrix}, & \text{if }  i \geq 2\,, 
    \end{cases} 
\end{aligned}  
\end{equation}
and 
\begin{equation} \label{eq:what-is-Bcl-i-model-based}
    B_{cl,1:i} := \begin{cases}
    B_1, & \text{if } i = 1\\
    \col(B_{cl,1:i-1},-\Upsilon_{c,i-1}B_{cl,1:i-1}), & \text{if } i \geq 2\,.
    \end{cases}   
\end{equation}  

\begin{theorem} \label{lemma:n-cascade-model-based}
Consider the cascade system \eqref{eq:n-system-cascade}. Define the coordinate transformation (for the $i$-th subsystem) as
\begin{equation} \label{eq:define-zeta-i}
    \zeta_{i}(k) := \begin{cases} 
    x_i(k),& \text{if } i = 1\\
    x_i(k) - \Upsilon_{c,i-1}\zeta_{1:i-1}(k), & \text{if } i \geq 2,
    \end{cases} 
\end{equation}
where $\Upsilon_{c,1}, \Upsilon_{c,2}, \dots \Upsilon_{c,\mathcal{N}-1}$ solve the nested Sylvester equations
\begin{equation} \label{eq:sylvester-N-systems}
\begin{cases} 
    \Upsilon_{c,1}A_{cl,1:1} = A_{2}\Upsilon_{c,1} + B_{2}\,, \\
    \Upsilon_{c,i}A_{cl,1:i} = A_{i+1}\Upsilon_{c,i} + B_{i+1}\begin{bmatrix}
        \Upsilon_{c,i-1} & I
    \end{bmatrix}\,, 
    \end{cases}       
\end{equation}
for all $2 \leq i \leq \mathcal{N} - 1$.
Then, the feedback control law \eqref{eq:control-n-cascade}
renders the cascade \eqref{eq:n-system-cascade} asymptotically stable for any $\mathcal{N} \geq 2$, provided that the matrices 
$N_i$, for all $i\in [1, \mathcal{N}]$, are such that
\begin{equation} \label{eq:n-cascade-gain-condition}
    \begin{cases}
    (A_1 + B_{cl, 1:1} N_1), & \text{if } i = 1\\
    (A_i - \Upsilon_{c, i-1} B_{cl, 1:i-1} N_i), & \text{if } i \geq 2\,, 
    \end{cases}   
\end{equation}
are all Schur, and $\lambda(A_{cl, 1:i}) \cap \lambda(A_{i+1}) = \emptyset$, for all $i\in [1, \mathcal{N} - 1]$. \\
\end{theorem}

\begin{proof} 
The claim is demonstrated by induction.

\textit{(Base case).} The two-cascade setting ($\mathcal{N}=2$) serves as a base-case, for which the statement holds, as recalled 
in Section~\ref{sec:cascade-problem-formulation}. Namely,  
the closed-loop dynamics  
are given by 
\eqref{eq:closed-loop-cascade-2-systems}, \ie{},
\begin{equation*}
    \zeta_{1:2}(k+1) = 
    \underbrace{\begin{bmatrix}
        A_1 + B_1 N_1 & B_1 N_2\\ \mathbf{0} & A_2 - \Upsilon_{c,1} B_1 N_2
    \end{bmatrix}}_{= A_{cl, 1:2}}
    \zeta_{1:2}(k)\,, 
\end{equation*}
where $A_{cl, 1:2}$ is Schur and, hence, the closed-loop 2-cascade system is asymptotically stable.

\textit{(Induction hypothesis).} Assume that the claim holds for $\mathcal{N} = j-1$ subsystems. Namely, the control input 
\begin{equation}\label{eq:proof-last-subsystem-input}
u_1(k) = \sum_{i=1}^{j-1} N_i \zeta_i(k),
\end{equation} 
with the feedback gains $N_1, N_2, \dots, N_{j-1}$ that satisfy \eqref{eq:n-cascade-gain-condition} yield the asymptotically stable closed-loop dynamics 
\begin{equation} \label{eq:proof-last-subsystem-closedloop}
     \zeta_{1:j-1}(k+1) = A_{cl, 1:j-1} \zeta_{1:j-1}(k),
\end{equation}
with $A_{cl, 1:j-1}$ Schur.

Then, as an intermediate step, note that if an additional input $\bar u(k)$ is added to the $(j-1)$-cascade system, \ie,  if  
\begin{equation}\label{eq:proof-last-subsystem-intermediate-input}
u_1(k) = \sum_{i=1}^{j-1} N_i \zeta_i(k) + \bar{u}(k),
\end{equation}  
\eqref{eq:proof-last-subsystem-closedloop} becomes  
\begin{equation} \label{eq:zeta_i_unforced}
    \zeta_{1:j-1}(k+1) = A_{cl, 1:j-1} \zeta_{1:j-1}(k) + B_{cl, 1:j-1} \bar{u}(k)\,. 
\end{equation}

\textit{(Induction step).} Now consider $\mathcal{N} = j$, namely an additional subsystem
\begin{equation} \label{eq:n-cascade-end-node}
    x_j(k+1) = A_j x_j(k) + B_j x_{j-1}(k),
\end{equation}
is added to the tail of the cascade considered in the induction hypothesis, and a new feedback term $v_j(k) = N_j \zeta_j(k)$ is introduced to the input $u_1(k)$  given in \eqref{eq:proof-last-subsystem-input}, yielding an overall control input of the form \eqref{eq:proof-last-subsystem-intermediate-input}, with 
$\bar{u}(k) = N_j \zeta_j(k)$. 
It can be observed that now the $\mathcal{N}$-cascade \eqref{eq:n-system-cascade} can be decomposed into a two-cascade setting described in Section \ref{sec:preliminaries-sylvester}, where $\mathbf{\Sigma_1}$ is system \eqref{eq:zeta_i_unforced}, which contains the first $(j-1)$ subsystems (under the coordinate transformation \eqref{eq:define-zeta-i}), and $\mathbf{\Sigma_2}$ is the new ``end subsystem'' 
\eqref{eq:n-cascade-end-node}. Note that by \eqref{eq:define-zeta-i}, we have that 
$x_{j-1} =\Upsilon_{c,j-2} \zeta_{1:j-2}(k) + \zeta_{j-1}(k)$,
yielding that 
$\mathbf{C_1} = \begin{bmatrix}
        \Upsilon_{c, j-2} & I
\end{bmatrix}$
for $\mathbf{\Sigma_1}$. Then, the Sylvester equation \eqref{eq:sylvester-N-systems}, for $i = \mathcal{N}$, is a direct specialisation of \eqref{eq:SylEqGeneric} to the setting where \eqref{eq:zeta_i_unforced} serves as $\mathbf{\Sigma_1}$ and \eqref{eq:n-cascade-end-node} serves as $\mathbf{\Sigma_2}$, therefore the dynamics of the ``end subsystem'' 
can be expressed as
\begin{equation} \label{eq:end-node-zeta-open-loop}
    \zeta_j(k+1) = A_j  \zeta_j(k+1) - \Upsilon_{c, j-1} B_{cl, 1:j-1} \bar{u}(k).
\end{equation}
Then, recalling that 
$\bar{u}(k) = N_j \zeta_j(k)$, 
it trivially follows from \eqref{eq:zeta_i_unforced} and \eqref{eq:end-node-zeta-open-loop} that the closed-loop of the entire cascade is described by
\begin{equation*}
    \zeta_{1:j}(k+1) = 
    \underbrace{\begin{bmatrix}
        A_{cl, 1:j-1} & B_{cl, 1:j-1} N_j \\ 
        \mathbf{0} & A_j - \Upsilon_{c, j-1} B_{cl, 1:j-1} N_j
    \end{bmatrix}}_{= A_{cl, 1:j}}
    \zeta_{1:j}(k).
\end{equation*}
Note that this is an 
upper-triangular matrix, $ A_{cl, 1:j-1}$ is Schur by the induction hypothesis, and $A_j - \Upsilon_{c, j-1} B_{cl, 1:j-1} N_j$  is Schur by the  condition \eqref{eq:n-cascade-gain-condition} on $N_j$, which implies that $A_{cl, 1:j}$ is Schur too. 
Asymptotic stability of the entire $\mathcal{N}$-cascade follows, which proves that the claim holds for $\mathcal{N} = j$. 

As the base case holds and the induction step has been proven under the induction hypothesis, by the principle of mathematical induction, the claim holds for any $\mathcal{N} \geq 2$.
\end{proof}

We now present a data-driven approach that solves the $\mathcal{N}$-cascade stabilisation problem by exploiting the result presented in Theorem~\ref{lemma:n-cascade-model-based}.

\begin{theorem} \label{theorem:n-cascade-data-driven}
Consider the cascade system \eqref{eq:n-system-cascade} and the control law \eqref{eq:control-n-cascade}. Suppose that Assumption~\ref{ass:rank-assumption} holds. Let 
\begin{equation}\label{eq:Ncasc_N1}
    N_1 := U_{1,-}Q_1(X_{1,-}Q_1)^{-1}\,, 
    \end{equation}
    be such that $\lambda(A_1 + B_1 N_1) \cap \lambda(A_2) = \emptyset$, where $Q_1 \in \RR^{T \times n_1}$ is any matrix that satisfies \eqref{eq:basic-sdp}, with $Q_1 := Q_x$. Let
    $\mathcal{A}_1 := X_{1, +} G_{N_1}$ and $\mathcal{B}_1 := X_{1, +} G_{B_1}$, with $G_{N_1} \in \RR^{T \times n_1}$ and $G_{B_1} \in \RR^{T \times n_0}$ any matrices that satisfy
\begin{equation} \label{eq:proof-base-AandB}
    \begin{bmatrix}
        I & \mathbf{0} \\ 
        N_1 & I
    \end{bmatrix}
    =\begin{bmatrix}
        X_{1,-} \\ U_{1,-}
    \end{bmatrix} 
    \begin{bmatrix}
        G_{N_1} & G_{B_1}
    \end{bmatrix},
\end{equation}
and let
$\mathcal{Y}_1 := X_{2,-} G_{\Upsilon_c}$, with $G_{\Upsilon_c} \in \RR^{T \times n_1}$ any matrix that satisfies \eqref{eq:cascade-sylvester-sdp}.\\
Let 
$$
\mathcal{A}_i :=
    \begin{bmatrix}
        \mathcal{A}_{i-1} & \mathcal{B}_{i-1} N_i \\ \mathbf{0} & A_i - \mathcal{Y}_{i-1} \mathcal{B}_{i-1} N_i
    \end{bmatrix}
$$ 
and 
$
\mathcal{B}_i := \col(\mathcal{B}_{i-1}, -\mathcal{Y}_{i-1}\mathcal{B}_{i-1}),
$
for all $i\in [2, \mathcal{N}]$, and let $\mathcal{Y}_{i} := X_{i+1, -} G_{\Upsilon_i}$, with $G_{\Upsilon_i} \in \RR^{T \times \sum_{j=1}^{i} n_j}$ any matrix that satisfies
\begin{subequations}  \label{eq:compute-ups-n-cascade}
\begin{empheq}[left=\empheqlbrace]{align}
    X_{i+1,+}  G_{\Upsilon_i} &= X_{i+1,-} G_{\Upsilon_i} \mathcal{A}_{i}   \\
    X_{i+1, -}  G_{\Upsilon_i} &= 
    \begin{bmatrix}
        \mathcal{Y}_{i-1} & I
    \end{bmatrix}
  \end{empheq}
\end{subequations}
provided that 
\begin{equation} \label{eq:rank-condition-upsilon-N-cascade}
\rank \left(\begin{bmatrix}
    X_{i+1,-}\\X_{i,-}
\end{bmatrix}\right) 
        = n_{i+1} + n_{i}\,,
\end{equation}
for all $i \in [2, \mathcal{N}-1]$.
For $i \in [2, \mathcal{N}]$, let\footnote{The notation $-/ +$ is used to indicate that the subscripts  may be either $-$ or $+$ (the meaning of which is defined in Section \ref{sec:introduction}).} 
\begin{equation*} 
    Z_{1:i,-/+} := \begin{cases} 
    X_{i, -/+}, & \text{if } i = 1\\
    \col(Z_{1:i-1,-/+}, Z_{i,-/+}), & \text{if } i \geq 2, 
    \end{cases} 
\end{equation*}
and let
\begin{equation}\label{eq:Ncascade_gains}
N_i := V_{i,-} Q_{i}(Z_{i,-} Q_{i})^{-1}
\end{equation} 
be such that $\lambda(A_{cl,1:i}) \cap \lambda(A_{i+1}) = \emptyset$, where $Q_i \in \RR^{T \times n_i}$ is any matrix that satisfies
\begin{equation}
        \label{eq:sdp-for-zeta-system-new}
    \begin{bmatrix} 
        Z_{i,-}Q_i & Z_{i, +}Q_i \\ Q_i^\top Z_{i, +}^\top & Z_{i,-}Q_i
    \end{bmatrix} \succ 0,
\end{equation}
with $Z_{i, -/+} := X_{i,-/+} - \mathcal{Y}_{i-1} Z_{1:i-1,-/+}$ and $V_{i,-} := U_{1,-} - \sum_{j=1}^{i-1} N_j Z_{j,-}$, provided that
\begin{equation} \label{eq:rank-condition-zeta-N-cascade}
    \rank \left(\begin{bmatrix}
            Z_{i,-}\\ V_{i, -}
    \end{bmatrix}\right) = n_{i} + n_0\,.
\end{equation}
Then, the feedback control law \eqref{eq:control-n-cascade}, with $N_1$ given by \eqref{eq:Ncasc_N1} and $N_i$ given by \eqref{eq:Ncascade_gains}, for all $i\in [2, \mathcal{N}]$, 
renders the cascade \eqref{eq:n-system-cascade} asymptotically stable for any $\mathcal{N} \geq 2$.
\end{theorem}

\begin{proof}
To prove the result, we show that the data-defined quantities $\mathcal{A}_i$, $\mathcal{B}_i$, and $\mathcal{Y}_i$ are the quantities $A_{cl,1:i}$, $B_{cl,1:i}$,  and $\Upsilon_{c,i}$, respectively. In addition, we show that the gains $N_i$ given by \eqref{eq:Ncasc_N1} and \eqref{eq:Ncascade_gains} are such that the matrices \eqref{eq:n-cascade-gain-condition} are Schur. Then the result follows from Theorem~\ref{lemma:n-cascade-model-based}. We split the proof into two cases.

\noindent \textit{($i = 1$)}. Note that
\begin{equation*}
\begin{aligned}
    \begin{bmatrix}
       \mathcal{A}_1 & \mathcal{B}_1
    \end{bmatrix} &= X_{1, +} \begin{bmatrix}
        G_{N_1} & G_{B_1}
    \end{bmatrix} \\
    &= (A_1 X_{1, -} + B_1 U_{1,-}) 
    \begin{bmatrix}
        G_{N_1} & G_{B_1}
    \end{bmatrix} \\
    &= \begin{bmatrix}
        A_1 & B_1
    \end{bmatrix}
    \begin{bmatrix}
        X_{1, -} \\
        U_{1,-}
    \end{bmatrix}
    \begin{bmatrix}
        G_{N_1} & G_{B_1}
    \end{bmatrix}  \\
    & = \begin{bmatrix}
        A_1 & B_1
    \end{bmatrix}
        \begin{bmatrix}
        I & \mathbf{0} \\ 
        N_1 & I
    \end{bmatrix}\\
    &= \begin{bmatrix}
      A_1 + B_1 N_1 & B_1
    \end{bmatrix}
    = \begin{bmatrix}
      A_{cl, 1:1} & B_{cl, 1:1}
    \end{bmatrix},
\end{aligned}
\end{equation*}
where in the fourth equality we used \eqref{eq:proof-base-AandB}. In addition, note that $A_1 + B_1 N_1$ (with $N_1$ given in \eqref{eq:Ncasc_N1}) is Schur (as shown in Theorem \ref{thm:modified-2-system-casc}), ensuring that the Schur condition on \eqref{eq:n-cascade-gain-condition} is satisfied for $i=1$.

\noindent \textit{($i \geq 2$)}. 
By specialising the Sylvester equation \eqref{eq:SylEqGeneric} in Lemma~\ref{lemma:computePI} to \eqref{eq:sylvester-N-systems}, it follows that \eqref{eq:compute-ups-n-cascade} is the corresponding data-dependent reformulation of the Sylvester equation \eqref{eq:sylvester-N-systems}, with the rank condition \eqref{eq:rank-condition-upsilon-N-cascade} ensuring that the corresponding analogue of Assumption \ref{ass:rankXU} holds. Therefore, we have that $\mathcal{Y}_i = \Upsilon_{c, i}$. Then, from
the given recursive constructions of $\mathcal{A}_i$ and $\mathcal{B}_i$, it follows directly that $\mathcal{A}_i = A_{cl, 1:i}$ and $\mathcal{B}_i = B_{cl, 1:i}$. 
Next, note that the transformed $i$-th system evolves according to the dynamics \eqref{eq:end-node-zeta-open-loop} (with $j:= i$) and that  $Z_{i,-/+}$ and $V_{i, -}$ are the corresponding data matrices of $\zeta_i$ and $\bar{u}$. Then, noting that the data-dependent LMI \eqref{eq:sdp-for-zeta-system-new} and the associated rank condition \eqref{eq:rank-condition-zeta-N-cascade} is an instance of the LMI \eqref{eq:basic-sdp} and the associated rank condition \eqref{eq:rank-condition}, it follows that the resulting $N_i$, given in \eqref{eq:Ncascade_gains}, renders $A_i - \Upsilon_{c, i-1} B_{cl,1:i-1}$ Schur, for any $i \geq 2$. \ie, \eqref{eq:n-cascade-gain-condition} is satisfied for all $i \in [2, \mathcal{N}]$.

\noindent Thus, the result follows from Theorem \ref{lemma:n-cascade-model-based}.
\end{proof}

\begin{remark} \label{rem:dataT}
If one disregards the cascade structure and attempts to stabilise the entire system as a monolithic entity via solving LMI \eqref{eq:basic-sdp}, then the state will be $\mathcal{X}(k):= \col(x_1(k),x_2(k),\cdots,x_{\mathcal{N}}(k)) \in \mathbb{R}^{\sum_{i=1}^{\mathcal{N}} n_i}$ and the required data length \(T\) would have to satisfy $T \ge T_{min}$ where
\begin{equation} \label{eq:lower-bound-data-savings-n-augmented}
    T_{min} = \sum_{i=0}^{\mathcal{N}} n_i,
\end{equation}
which, assuming that the individual subsystem orders $n_i$ are 
independent of $\mathcal{N}$, is $O(\mathcal{N})$. In contrast, when exploiting the cascade structure as in Theorem~\ref{theorem:n-cascade-data-driven}, the given rank conditions only require that the minimum data length is
\begin{equation} \label{eq:lower-bound-data-savings-n-forwarding}
    T_{min} = \max_{1 \leq i \leq \mathcal{N}-1}\!\bigl(n_{\mathcal{N}} + n_0,\,  n_i + n_0,\, n_{i+1} + n_{i}\bigr),
\end{equation}
which is $O(1)$ in the number of cascade stages $\mathcal{N}$, under the same assumption, showing a significant improvement in terms of data efficiency.
\end{remark}

\subsection{Numerical Example}
\label{sec:cascade-example}

\begin{figure}
    \centering
    \includegraphics[width=0.9\columnwidth]{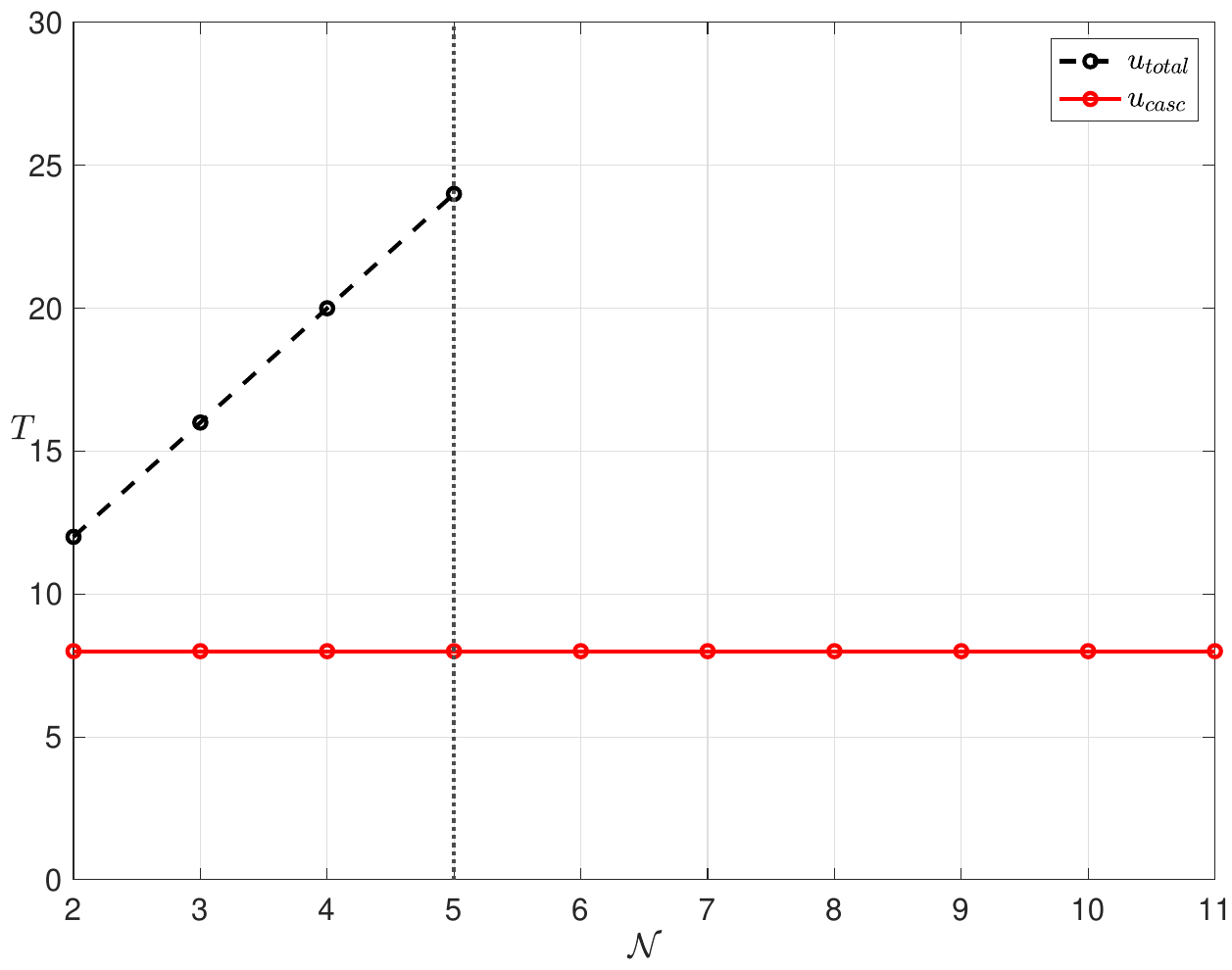}
    \caption{$T$ required to stabilise a cascade of $\mathcal{N}$ subsystems for both $u_{casc}$ (solid/red line) and $u_{total}$ (dashed/black line). The vertical dotted/black line indicates the boundary, to the right-hand side of which the approach $u_{total}$ fails to stabilise the entire cascade.} 
    \label{fig:n_cascades_min_datapoints}
\end{figure}%

To demonstrate the results of the previous sections, we consider an interconnection of $\mathcal{N}$ subsystems, 
each with state $x_i(k) \in \mathbb{R}^{4}$ for all $i \in [1,\mathcal{N}]$, and input $u_1(k) \in \mathbb{R}^4$.  More precisely, consider two 
discrete-time systems defined by (the randomly generated) matrix 
pairs $[\mathbf{A}_1|\mathbf{B}_1]$ and $[\mathbf{A}_2|\mathbf{B}_2]$, as 
given in \eqref{eq:numerical-cascade-first-system} and \eqref{eq:numerical-cascade-second-system}, respectively, and consider the 
cascade system \eqref{eq:n-system-cascade} 
constructed with
\begin{equation}
    [A_i|B_i] = \begin{cases}
        [\mathbf{A}_1|\mathbf{B}_1], & \text{if } i \text{ is odd,}\\
        [\mathbf{A}_2|\mathbf{B}_2], & \text{if } i \text{ is even,}\\
    \end{cases}
\end{equation}
for $i \in [1,\mathcal{N}]$. The matrices $[\mathbf{A}_1|\mathbf{B}_1]$ and $[\mathbf{A}_2|\mathbf{B}_2]$ were generated using MATLAB's \textit{rss} function for random continuous-time stable systems, and subsequently modified to shift one eigenvalue of each matrix $A_i$ onto the right half plane. Finally, the systems were discretised using a zero-order hold with $t_s=0.1$~s. Thus, each $A_i$ is such that, in the absence of any control input, the subsystems are unstable (as is the whole cascade when viewed as a monolithic entity).

\begin{figure*}[t]  
\centering
\begin{equation}\label{eq:numerical-cascade-first-system}
[\mathbf{A}_1|\mathbf{B}_1] =  
\left[
\begin{array}{rrrr|rrrr}
0.7024  &  0.1639  & -0.5026  &  0.1318 & 0.1121  & -0.0274  &  0.1904  &  0.0153\\
   -0.1507  &  0.9198  &  0.1836  & -0.0095 &0.0921  & -0.0322  &  0.0756  & -0.0028\\
    0.4880  &  0.0151   & 0.7407  &  0.2593 & 0.0439  &  0.0033  & -0.0406  & -0.0165\\
   -0.1897   & 0.1716  & -0.1388   & 0.7960 &  -0.0485 &    0.0914  &  0.0949  & -0.0124\\
\end{array}
\right]
\end{equation}
\end{figure*}

\begin{figure*}[t]  
\centering
\begin{equation}\label{eq:numerical-cascade-second-system}
[\mathbf{A}_2|\mathbf{B}_2] = 
\left[
\begin{array}{rrrr|rrrr}
0.7426  & -0.1530  &  0.2112  &  0.5858 & -0.1977 &  -0.0318  & -0.1206  & -0.0121\\
0.2451  &  0.9275  &  0.0160  & -0.0617 & 0.0034  & -0.0272  &  0.0927  &  0.1312\\
-0.3498  &  0.0315  &  0.9548  &  0.0431 & -0.3043 &   0.0372  &  0.1062   & 0.0012\\
-0.4782  &  0.1994  & -0.2835  &  0.7992 & -0.2300  & -0.1325  &  -0.1445  &  0.0127\\
\end{array}
\right]
\end{equation}
\end{figure*}

In what follows we consider several cases, which correspond to different values of $\mathcal{N}\ge 2$. For each case, we consider two different data-driven approaches to stabilise the cascade. 
First we apply the proposed data-driven forwarding procedure (in Theorem \ref{theorem:n-cascade-data-driven}) resulting in a control input which we denote by 
$u_{casc}$. 
Second, we apply
the approach in \cite{de2019formulas} directly, effectively treating 
the cascade as a monolithic system, resulting in a control input denoted by 
$u_{total}$. 
In each of these cases, we generate persistently exciting random inputs using MATLAB's \textit{randn} function, in order to create data sequences of set length of at least $T_{min}$ as defined in \eqref{eq:lower-bound-data-savings-n-augmented} and \eqref{eq:lower-bound-data-savings-n-forwarding} for the monolithic and the cascade control architectures,
respectively. In all simulations in this section, the rank conditions have been checked and confirmed to be always satisfied. 
All convex problems are solved using MATLAB CVX \cite{cvx}, with MOSEK \cite{aps2019mosek} as the selected solver for LMI problems and the default SDPT3 solver for finding the data-based Sylvester solutions. For each method, if stabilisation was unsuccessful (\ie, the solver failed to find a solution to the optimisation problem despite the rank conditions being satisfied) using data sequences of length $T = T_{min}$, then the sequence length $T$ was incremented until either stabilisation occurred or $T$ became greater than $2T_{min}$, at which point it was assumed that the method could not stabilise the cascade. It should be noted that theoretically using $T > T_{min}$ should have no effect, but this strategy was tested to exclude that numerical artifacts were causing the solver to fail.

\begin{figure}
    \centering
    \includegraphics[width=0.95\columnwidth]{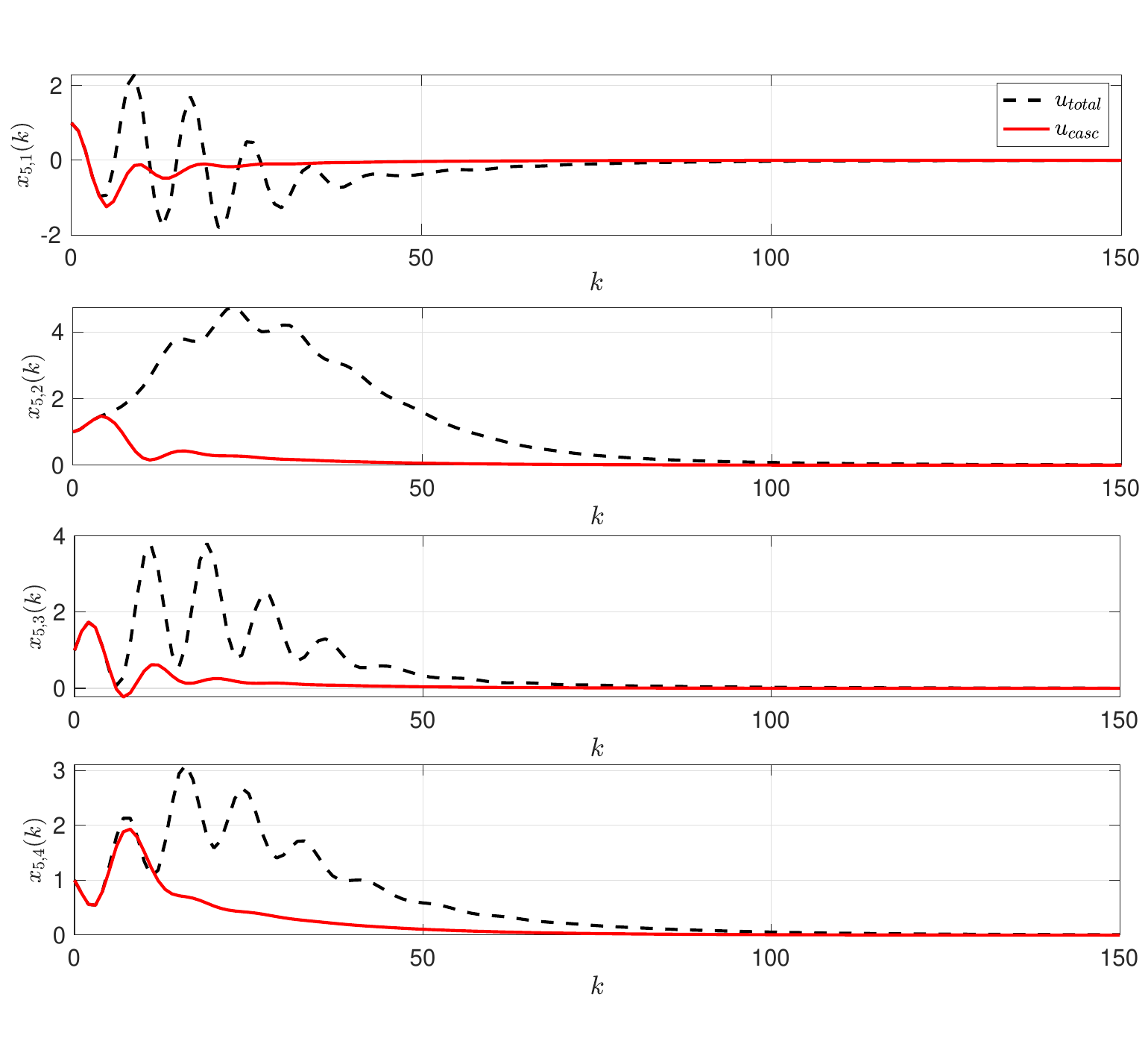}
    \caption{Time histories of the state of the end subsystem (with $\mathcal{N} = 5$) for $u_{casc}$ (solid/red line), and $u_{total}$ (dashed/black line). Each subplot corresponds to a state variable.}
    \label{fig:n_cascades_traj}
\end{figure}

\begin{figure}
    \centering
    \includegraphics[width=0.95\columnwidth]{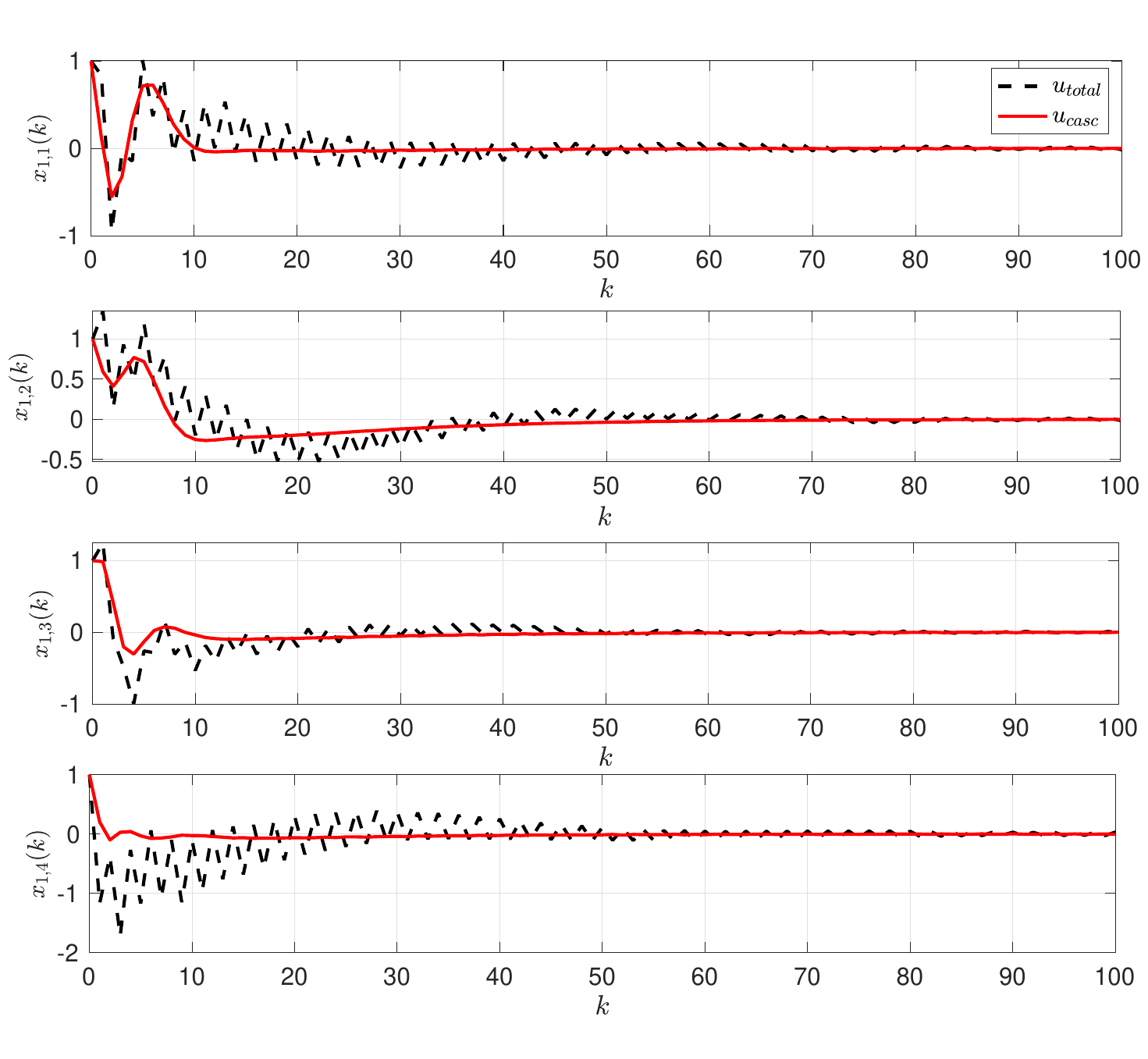}
    \caption{Time histories of the state of the first subsystem (with $\mathcal{N} = 2$) for $u_{casc}$ (solid/red line), and $u_{total}$ (dashed/black line), in the presence of measurement noise on both subsystems bounded by $|\Delta x_1| \le 0.001$ and $|\Delta x_2| \le 0.001$.}
    \label{fig:noise_traj}
\end{figure}

In Fig.~\ref{fig:n_cascades_min_datapoints}, we compare, for different values of $\mathcal{N}$, the minimum number of data points (or equivalently, the minimum time length) $T$ required to stabilise a cascade of $\mathcal{N}$ subsystems for both approaches. It was observed that the control strategy $u_{casc}$ stabilises cascades of subsystems up to $\mathcal{N}=11$, requiring only $8$ data points, in accordance with the value of $T_{min}$ given in \eqref{eq:lower-bound-data-savings-n-forwarding}, which is independent of $\mathcal{N}$. Meanwhile, for $u_{total}$, $T$ scales (linearly) with $\mathcal{N}$ in accordance with $T_{min}$ given in \eqref{eq:lower-bound-data-savings-n-augmented}. We note that for $\mathcal{N}=5$, the monolithic method requires three times as many data points as our proposed method. Moreover, the monolithic method cannot stabilise the cascade when $\mathcal{N} > 5$, as indicated by the vertical dotted/black line, meaning that the LMI solver fails to find a solution. Thus, Fig. \ref{fig:n_cascades_min_datapoints} shows that the bounds given in Remark~\ref{rem:dataT} are sharp and demonstrates that the data-driven forwarding method $u_{casc}$ outperforms $u_{total}$ in two ways: it requires fewer data points for each $\mathcal{N}$ tested, and can stabilise cascades of sizes for which it is not possible to determine $u_{total}$. 

We recall that in Remark \ref{remark:rank-condition-cascade-2} we observed that, since
during the data generation experiment it is only possible to directly influence the input data sequence $U_{1,-}$, it may be difficult to guarantee
that the rank conditions \eqref{eq:rank-condition-upsilon-N-cascade} and \eqref{eq:rank-condition-zeta-N-cascade} will be satisfied from the generated data sequences, as these do not depend explicitly on $U_{1,-}$. However, in this numerical example it was found that these rank conditions were easily satisfied with the random input sequence supplied, as long as the data sequence length was selected as $T \ge T_{min}$. Thus, it is unlikely that these rank conditions not being guaranteed would be a barrier to the adoption of the forwarding method in practice.

Fig.~\ref{fig:n_cascades_traj} shows the time histories of the state of the end subsystem with $\mathcal{N} = 5$, for both $u_{casc}$ and $u_{total}$, to visually show that the state responses are comparable, although the monolithic method generates larger oscillations.

We now provide some empirical results with noisy data, where for simplicity we restrict our focus to $\mathcal{N} = 2$, \ie, the standard 2-cascade setting. In this example, the generated data matrices are corrupted by bounded measurement noise $\Delta x_1$ and $\Delta x_2$. 
For $u_{casc}$, this was achieved by generating data sequences of noise in the form $\Delta X_{1,-/+}$ and $\Delta X_{2,-/+}$ using MATLAB's \textit{rand} function, and then by capping these numbers between $\pm 10^{-i}$ for $i \in [1,8]$ (\ie, we tested eight levels of noise). These noise data sequences were added to the noise-free state data sequences to obtain noise-corrupted data sequences $\bar{X}_{1,-/+}$ and $\bar{X}_{2,-/+}$. This same process was repeated for $u_{total}$. It was found that both $u_{casc}$ and $u_{total}$ showed similar performance: both methods stabilised for a $\Delta x$ bounded up to $\pm 10^{-3}$, despite the SNR conditions \eqref{eq:snr-condition-normal} and \eqref{eq:SNR-condtion-zeta} were typically satisfied only up to a $\Delta x$ bounded by $10^{-4}$. This confirms that the conditions \eqref{eq:snr-condition-normal} and \eqref{eq:SNR-condtion-zeta} are indeed conservative (as shown also in \cite[Section V]{de2019formulas}), and shows that the forwarding method has similar robustness to standard data driven stabilisation with noise. Fig.~\ref{fig:noise_traj} shows the time history of the state of the first subsystem for both $u_{casc}$ (solid/red line) and $u_{total}$ (dashed/black line) when subject to noise bounded by $\pm10^{-3}$. For this simulation, the norm $\|\Delta \Upsilon_c\|_2$ was found to have a value of $1.075$, well within the expected bound of $46.364$ calculated according to Lemma~\ref{lemma:errorBound-cascade}. Thus, in this specific instance the bound is loose (with a ratio of around $40$). However, the tightness of this bound varies greatly due to the random nature of the simulation. In other cases, the ratio was found to have tighter ratios of $\sim\!\!12$.

\section{Conclusion} \label{sec:conclusion}
Motivated by the well-known relation between Sylvester equations and cascade interconnections, we have developed a data-driven framework that solves the Sylvester equation from data samples, and we have provided an error analysis under noise. This framework is applicable to a broad class of problems that can be interpreted through system interconnections. We have showcased the applicability and effectiveness of the proposed framework by addressing the (data-driven) cascade stabilisation problem across a variety of settings. In Part II \cite{mao2025partTwo} of this article, we will exploit this framework further to provide direct data-driven solutions to the model order reduction and output regulation problems.

\section*{Acknowledgment}
G. Scarciotti would like to thank the members of the ``Sylvester Kings'' group, Daniele Astolfi and John Simpson-Porco, for daily illuminating discussions over several years regarding control theoretic implications of the Sylvester equation.

\section*{References}
\bibliographystyle{IEEEtran}
\bibliography{ref}

\begin{IEEEbiography}[{\includegraphics[width=1in,height=1.25in,clip,keepaspectratio]{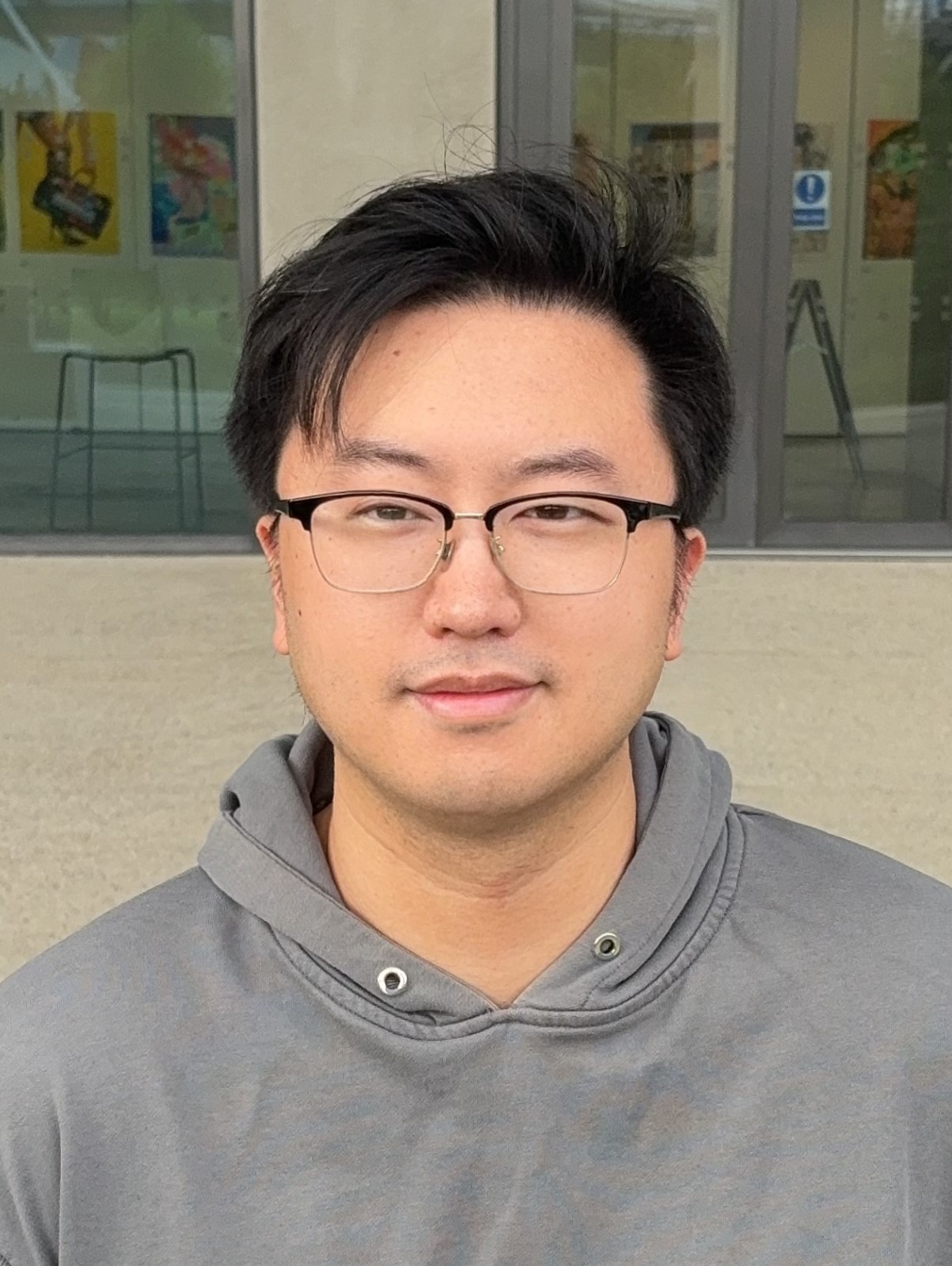}}]{Junyu Mao} (Student Member, IEEE) received his B.Eng. degree in Electrical and Electronic Engineering from the University of Liverpool, Liverpool, U.K., in 2018, and two M.Sc. degrees---in Control Systems from Imperial College London, London, U.K., in 2019, and in Data Science and Machine Learning from University College London, London, U.K., in 2020. He is currently pursuing the Ph.D. degree in the Control and Power Group at Imperial College London. In December 2024, he was a visiting Ph.D. student at the French National Centre for Scientific Research (CNRS). His research interests include the theoretical foundations of model order reduction and data-driven control for large‑scale dynamical systems.
\end{IEEEbiography}

\begin{IEEEbiography}[{\includegraphics[width=1in,height=1.25in,clip,keepaspectratio]{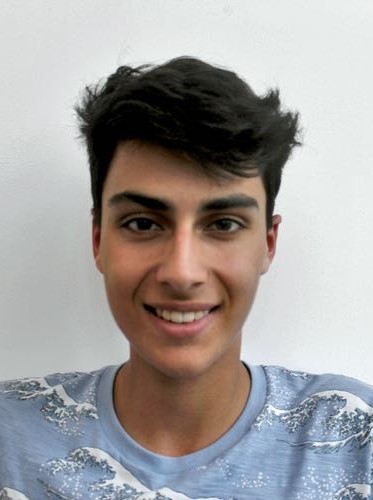}}]{Emyr Williams} 
is a Ph.D Candidate at the MAC-X Lab, Imperial College London. He received his M.Eng. degree in Aeronautical Engineering from Imperial College London in 2022. Following two years in industry as a software engineer, he joined the Control and Power research group at Imperial College London in 2024 to pursue his Ph.D. His research specialises on the application of data-driven and nonlinear control methodologies in the offshore renewable sector.
\end{IEEEbiography}

\begin{IEEEbiography}[{\includegraphics[width=1in,height=1.25in,clip,keepaspectratio]{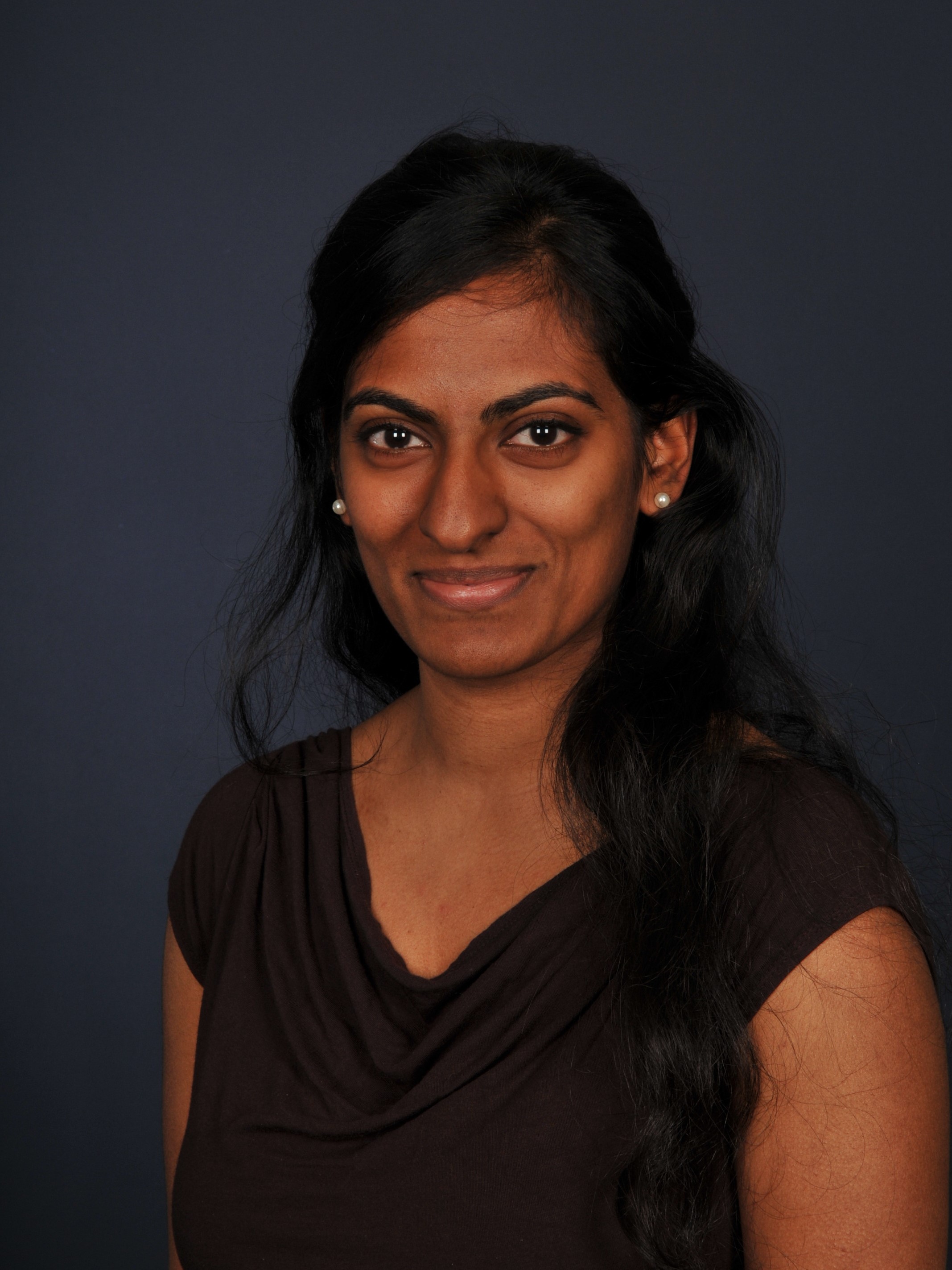}}]{Thulasi Mylvaganam} was born in Bergen, Norway, in 1988. She received the M.Eng. degree in Electrical and Electronic Engineering and the Ph.D. degree in Nonlinear Control and Differential Games from Imperial College London, London, U.K., in 2010 and 2014, respectively. From 2014 to 2016, she was a Research Associate with the Department of Electrical and Electronic Engineering, Imperial College London. From 2016 to 2017, she was a Research Fellow with the Department of Aeronautics, Imperial College London, UK, where she is currently Associate Professor. Her research interests include nonlinear control, optimal control, game theory, distributed control and data-driven control. She is Associate Editor of the IEEE Control Systems Letters, of the European Journal of Control, of the IEEE CSS Conference Editorial Board and of the EUCA Conference Editorial Board. She is also Vice Chair of Education for the IFAC Technical Committee 2.4 (Optimal Control) and Member of the UKACC Executive Committee.
\end{IEEEbiography}

\begin{IEEEbiography}[{\includegraphics[width=1in,height=1.25in,clip,keepaspectratio]{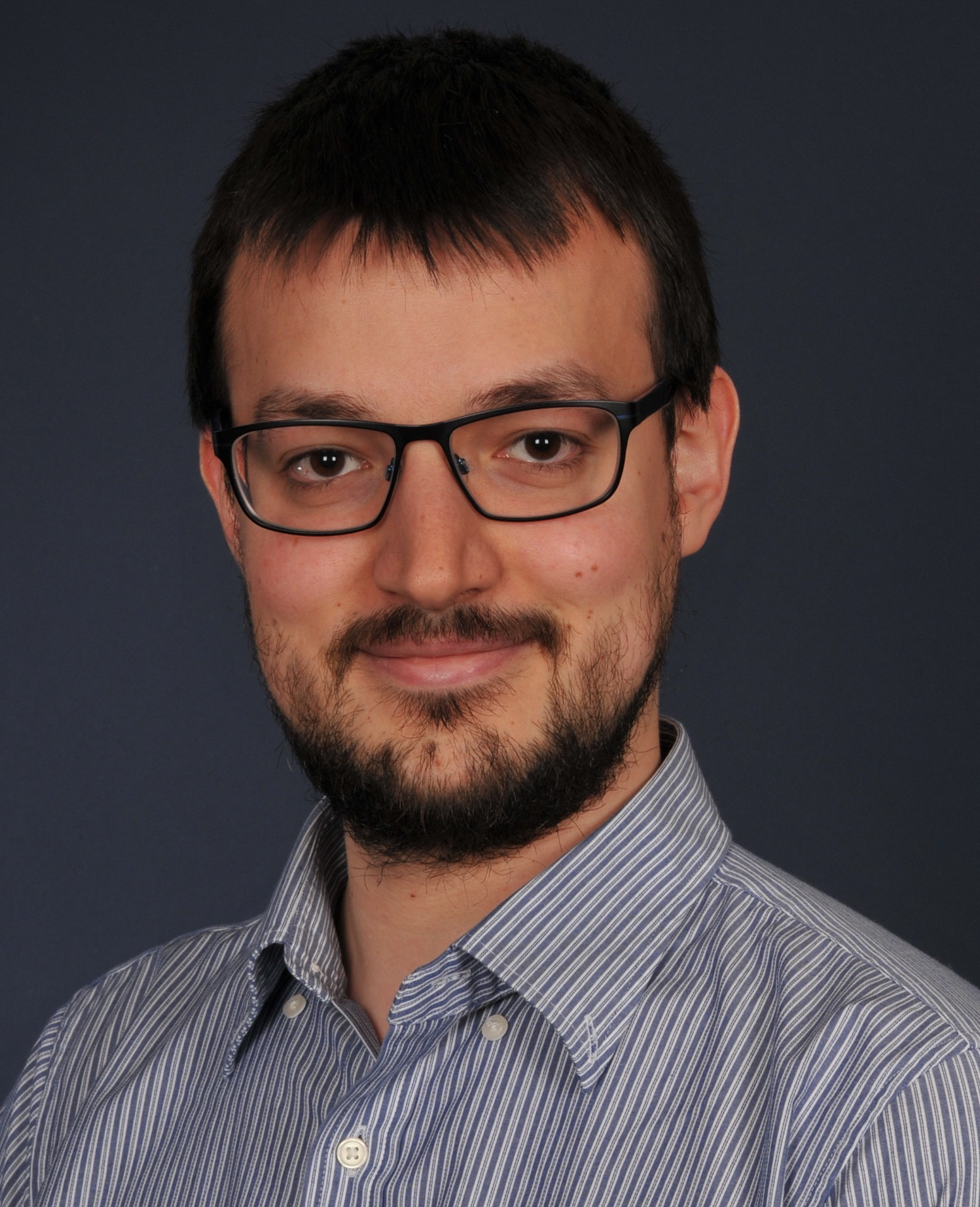}}]{Giordano Scarciotti} (Senior Member, IEEE) received his B.Sc. and M.Sc. degrees in Automation Engineering from the University of Rome ``Tor Vergata'', Italy, in 2010 and 2012, respectively. In 2012 he joined the Control and Power Group, Imperial College London, UK, where he obtained a Ph.D. degree in 2016. 
He also received an M.Sc. in Applied Mathematics from Imperial in 2020. He is currently an Associate Professor at Imperial. 
He was a visiting scholar at New York University in 2015, at University of California Santa Barbara in 2016, and a Visiting Fellow of Shanghai University in 2021-2022. He is the recipient of an Imperial College Junior Research Fellowship (2016), of the IET Control \& Automation PhD Award (2016), the Eryl Cadwaladr Davies Prize (2017), an ItalyMadeMe award (2017) and the IEEE Transactions on Control Systems Technology Outstanding Paper Award (2023). He is a member of the EUCA Conference Editorial Board, of the IFAC and IEEE CSS Technical Committees on Nonlinear Control Systems and has served in the International Programme Committees of multiple conferences. He is Associate Editor of Automatica. He was the National Organising Committee Chair for the EUCA European Control Conference (ECC) 2022, and of the 7th IFAC Conference on Analysis and Control of Nonlinear Dynamics and Chaos 2024, and the Invited Session Chair and Editor for the IFAC Symposium on Nonlinear Control Systems 2022 and 2025, respectively. He is the General Co-Chair of ECC 2029. 
\end{IEEEbiography}

\end{document}